\documentclass[a4paper,10pt]{article}

\usepackage[pages=all, color=black, position={current page.south}, placement=bottom, scale=1, opacity=1, vshift=5mm]{background}
\SetBgContents{
} 
\usepackage[margin=0.92in]{geometry} 

\usepackage{amsmath}
\usepackage{amsthm}
\usepackage{amssymb}

\usepackage[utf8]{inputenc}
\usepackage{hyperref}
\hypersetup{
	unicode,
	pdfauthor={Virgile Guemard},
	pdftitle={Lifting a CSS code via its handlebody realization},
	pdfsubject={A simple article template},
	pdfkeywords={article, template, simple},
	pdfproducer={LaTeX},
	pdfcreator={pdflatex}
}

\usepackage[sort&compress,numbers,square]{natbib}

\theoremstyle{plain}
\newtheorem{theorem}{Theorem}[section]
\newtheorem{corollary}[theorem]{Corollary}
\newtheorem{lemma}[theorem]{Lemma}

\newtheorem{proposition}[theorem]{Proposition}

\theoremstyle{definition}
\newtheorem{definition}[theorem]{Definition}
\newtheorem{example}[theorem]{Example}
\newtheorem{remark}[theorem]{Remark}

\usepackage{graphicx, color}
\graphicspath{{fig/}}

\usepackage{algorithm, algpseudocode} 
\usepackage{mathrsfs} 

\title{Lifting a CSS code via its handlebody realization}
\author{Virgile Guemard$^{1,2}$}

\date{
	$^1$Aix Marseille Université, I2M, UMR 7373, 13453 Marseille, France\\%
        $^2$Inria Paris, France \\[2ex]%
	\today
}

\usepackage{dcolumn}
\usepackage{bm}
\usepackage[english]{babel}
\usepackage[OT1]{fontenc}

\usepackage[colorinlistoftodos, color=green!40, prependcaption]{todonotes}

\usepackage{amsthm}
\usepackage{mathtools}
\usepackage{physics}
\usepackage{xcolor}
\usepackage{graphicx}
\usepackage{tikz-cd}
\usepackage{bbold}

\usepackage{adjustbox}
\usepackage{placeins}

\usepackage{csquotes}
\usepackage{booktabs}

\usepackage[caption=false]{subfig}
 \makeatletter
\renewcommand{\fnum@figure}{FIG. \thefigure}
\makeatother
\usepackage{tikz}
\usetikzlibrary{decorations.markings}

\usepackage{amsmath,mleftright}
\usepackage{xparse}

\NewDocumentCommand{\evalat}{sO{\big}mm}{%
  \IfBooleanTF{#1}
   {\mleft. #3 \mright|_{#4}}
   {#3#2|_{#4}}%
}

\usepackage{amsthm}

\usepackage{extpfeil}
\usepackage{tikz-cd}
\usepackage{svg}

\usepackage{titlesec}
\titleformat{\section}
  {\normalfont\bfseries}{\thesection}{1em}{}
\titleformat{\subsection}
  {\normalfont\bf}{\thesubsection}{1em}{}{\normalfont\bfseries}
\titleformat{\subsubsection}
  {\normalfont\bfseries}{\thesubsubsection}{1em}{}
\renewcommand{\thesubsection}{\thesection.\arabic{subsection}}
\renewcommand{\thesection}{\arabic{section}}
\renewcommand{\thesubsubsection}{\thesubsection.\arabic{subsubsection}} 

\usetikzlibrary{arrows}

\begin{document}
	\maketitle

\begin{abstract}
We present a topological approach to lifting a quantum CSS code. In previous work \cite{GuemardLiftIEEE}, we proposed lifting a CSS code by constructing covering spaces over its 2D simplicial complex representation, known as the Tanner cone-complex. This idea was inspired by the work of Freedman and Hastings \cite{Freedman2020CSS_Manifold}, which associates CSS codes with handlebodies. In this paper, we show how the handlebody realization of a code can also be used to perform code lifting, and we provide a more detailed discussion of why this is essentially equivalent to the Tanner cone-complex approach. As an application, we classify lifts of hypergraph-product codes via their handlebody realization.
\end{abstract}

\tableofcontents

\section{Introduction}

Quantum error-correcting codes are believed to be a necessary component of most large-scale quantum computing architectures. Among the wide variety of error correction schemes, Calderbank–Shor–Steane (CSS) codes~\cite{Calderbank1996,Shor1995} have attracted considerable attention since their discovery, due to their simple mathematical formulations~\cite{Freedman2002,Kitaev2003}, namely chain complexes, which have been utilized to develop new topological and combinatorial code constructions, as well as protocols for their practical implementation.\par 

In recent years, major progress has been made in the construction of asymptotically good low-density parity-check (LDPC) CSS codes, which are CSS codes defined by a pair of sparse parity-check matrices. This began with the geometrical intuition of Hastings, Haah, and O'Donnell~\cite{Hastings2020}, that the distance of a hypergraph-product code (HPC), a code defined by the tensor product of two linear codes, could be increased by twisting one of the two codes along the other. The formalization of this idea, analogous to a fiber bundle, led to significant progress in \cite{Panteleev2020}, where the distance was shown to grow almost linearly. Soon after, \cite{Breuckmann2021} proposed that the tensor product of Sipser-Spielman codes over non-Abelian Cayley graphs could yield a family of codes with linear distance. This eventually resulted in a family of LDPC codes, called lifted product codes (LPC) \cite{Panteleev2021}, whose dimension and distance scale linearly with the code length, thus referred to as \textit{good} codes. This construction was later simplified in \cite{Leverrier2022,leverrier2022decoding}, where a connection to the classical locally testable code family of \cite{Dinur2021} was also established.
\par

This progress in quantum coding theory has had a positive impact on geometry, notably through the work of Freedman and Hastings~\cite{Freedman2020CSS_Manifold}. Using LDPC code families with distance scaling as $n^\alpha$, for $\alpha > 1/2$, they proved the existence of a family of 11-dimensional closed, compact Riemannian manifolds with power-law $\mathbb{Z}_2$-systolic freedom. In their work, they present an explicit procedure that takes as input any CSS code, i.e. a length-3 chain complex with $\mathbb{Z}_2$ coefficients, and constructs a handlebody such that the code appears as a subcomplex of its cellular chain complex.\par 

This manifold construction served as inspiration in~\cite{GuemardLiftIEEE} to define a general notion of code lifting for quantum CSS codes: it is indeed always possible to generate a covering space over the manifold realization of a code and extract a new code from it. However, while this intuition was used as a guide to construct and analyze new codes, it was found that a simpler way to generate a lifted code is to create a covering of a newly discovered simplicial complex, called the Tanner cone-complex\footnote{Nevertheless, the present note was essentially written prior to~\cite{GuemardLiftIEEE}.}. In particular, this avoids a cumbersome step from~\cite{Freedman2020CSS_Manifold} referred to as the $\mathbb{Z}$-lift of a chain complex.\par

The present work aims to clarify the parallel between the Freedman–Hastings manifold and the Tanner cone-complex of a code, as defined in~\cite{GuemardLiftIEEE}. This correspondence naturally explains the relationship between a covering of the manifold and a lift of the code. In particular, for certain families of codes, generating a covering of the manifold associated with a code is equivalent, although harder to implement, to the code lifting of~\cite{GuemardLiftIEEE}.\par

Our motivation is also to translate the code lifting principle into the language of topological codes. Indeed, recent work motivates further investigation of topological codes and the Freedman–Hastings construction, as in~\cite{portnoy23,zhu2024noncliffordparallelizablefaulttolerantlogical,zhu2025topologicaltheoryqldpcnonclifford}, where new codes have been developed under additional constraints, such as geometric locality or the existence of transversal non-Clifford logical gates based on the cup product. Another application is the efficient computation of sparse $\mathbb{Z}$-lifts for certain codes. Let $C'$ be a lift, in the sense of~\cite{GuemardLiftIEEE}, of a CSS code $C$ such that $C$ admits a support-preserving $\mathbb{Z}$-lift, as defined in Section \ref{section cellular-lift of a CSS code}. Then, our procedure yields a support-preserving $\mathbb{Z}$-lift of $C'$.\par 

This article is organized as follows. In Section~\ref{section Preliminaries}, we recall essential definitions related to CSS codes and code lifting. Section~\ref{section the freedman hastings construction} reviews the CSS-to-manifold mapping introduced by Freedman and Hastings. In Section~\ref{section cellular-lift of a CSS code}, we provide an explicit code lifting procedure based on this manifold mapping. We also elaborate on the connection between this framework and the code lifting technique proposed in~\cite{GuemardLiftIEEE}. Finally, we show how to adapt the classification of LPCs as lifts of HPCs using topological arguments.

\section{Preliminaries}\label{section Preliminaries}

\subsection{Quantum CSS codes}\label{section quantum CSS codes}

CSS codes are a class of stabilizer quantum error-correcting codes. The principles first appeared in the seminal works of Calderbank, Shor, and Steane~\cite{Calderbank1996,Steane1996,Stean1996Multiple}. These codes are defined by a pair of classical linear codes, or equivalently by two parity-check matrices that satisfy a specific orthogonality condition. Soon after, it was recognized that this condition could be formulated in the language of chain complexes~\cite{Kitaev2003,Freedman2001}, opening new perspectives for code construction and analysis; see~\cite{Audoux2015,Breuckmann2021review} for a detailed review.

\begin{definition}\label{Def_CSS stabilizer code}
Let $C_X$ and $C_Z$ be two linear codes with parity-check matrices $H_X$ and $H_Z$, such that $C_X^\perp\subseteq C_Z \subset \mathbb{F}_2^{n}$. A $[[ n, k, d]]$ CSS code is the subspace of $(\mathbb{C}_2)^{\otimes n}$ given by\[\text{CSS}(C_X,C_Z):=\operatorname{Span}\left \{\sum_{u\in C_Z^\perp }\ket{c+u}\text{ }|\text{ } c\in C_X \right \}.\]
Its parameter are:
\begin{itemize}
\item the \textit{code length} $n$,
\item the \textit{dimension} $k = \dim(C_X / C_Z^\perp)=n-\operatorname{rank} H_X-\operatorname{rank} H_Z$,
\item the \textit{distance} $d = \min(d_X, d_Z)$, with the $X$ and $Z$ distances,
\begin{align*}
d_X=\min\limits_{c\in C_Z\setminus C_X^\perp}|c|,\qquad d_Z=\min\limits_{c\in C_X\setminus C_Z^\perp}|c|.
\end{align*}
\end{itemize}
A family of CSS codes is said $good$ when its parameters are, asymptotically, $[[ n, \Theta(n), \Theta(n)]] $.  A family is called LDPC if $H_X$ and $H_Z$ are sparse matrices, i.e. rows and columns have weight bounded by a constant.
\end{definition}
The dimension corresponds to the number of encoded logical qubits, while $\lfloor (d-1)/2 \rfloor$ is the number of correctable errors. A high-level objective is to design quantum CSS codes with the largest possible dimension and distance for a given number of physical qubits $n$.\par

Physically, the rows of $H_X$ define the $X$-type parity-check operators, or $X$-checks, and similarly, the rows of $H_Z$ define the $Z$-checks. The orthogonality condition between $C_X$ and $C_Z$, equivalent to $H_X H_Z^T = 0$, is necessary to ensure that the syndromes of the two codes can be extracted independently, i.e., via commuting operators. For simplicity, we denote the sets of $Z$- and $X$-checks\footnote{In this work, by a $Z$-check, we always mean a generator corresponding to a row of $H_Z$, and likewise for $X$-checks and $H_X$.} as $Z$ and $X$, respectively, corresponding to the rows of $H_Z$ and $H_X$. We let $Q$ denote the set of qubits, corresponding to the columns of $H_X$ and $H_Z$. Thus, $|Q| = n$ is the number of qubits, while $|X|$ and $|Z|$ denote the number of rows of $H_X$ and $H_Z$, respectively.\par

The relation $H_X H_Z^T = 0$, known as \textit{the orthogonality condition}, is analogous to the composition property of two boundary maps in a chain complex. Since the data of two parity-check matrices is sufficient to construct a CSS code, it is naturally defined as a \textit{3-term chain complex} or its dual.

\begin{proposition}
A CSS code $C=\operatorname{CSS}(C_X,C_Z)$, defined by a pair of parity check matrices $H_X, H_Z$, is equivalent to a 3-term complex given with a basis, $C:=
C_{i+1}\xrightarrow{\partial_{i+1}}C_i\xrightarrow{\partial_i}C_{i-1}=\mathbb{F}_2^{|Z|}\xrightarrow{H_Z^T}\mathbb{F}_2^{|Q|}\xrightarrow{H_X} \mathbb{F}_2^{|X|}$, where $i\in \mathbb Z$. Its parameters $ [[ n, k, d]]$ are  
\begin{align*}
    n=&\dim(C_i),\\
    k=&\dim(H_i(C))=\dim(H^i(C)),\\
    d=&\min \left \{|c|: [c]\in H_i(C)\sqcup H^i(C), [c]\neq0\right \},
\end{align*}
where $H_i(C)$ denotes the $i$-th homology group of $C$.
\end{proposition}
Assuming that a basis is fixed for each vector space $C_i$, the boundary maps can be interpreted as matrices. They play the role of the parity-check matrices $H_X$ and $H_Z$ or their transpose, and the linear codes are the subspaces $C_X = \operatorname{ker}(\partial_i)$ and $C^{\perp}_Z = \operatorname{Im}(\partial_{i+1})$.\par

For code constructions known as topological codes~\cite{Kitaev2003}, it has been shown to be useful to extract quantum CSS codes from any based chain complex. Indeed, a chain complex of length greater than three can be truncated into a 3-term one. This aspect will also be used later in this note to extract a code from a cell complex of dimension 5.\par

It will also be appropriate to define the chain complex associated with a code in terms of abstract cells taken directly from the sets of checks and qubits, such as $C = \mathbb{F}_2 Z \xrightarrow[]{\partial_2} \mathbb{F}_2 Q \xrightarrow[]{\partial_1} \mathbb{F}_2 X$, where $\mathbb{F}_2 S = \bigoplus_{s \in S} \mathbb{F}_2 s$ denotes the formal linear combination, called \textit{chains}, of abstract basis cells in the sets $S = X, Q$, or $Z$. A boundary map is defined by $\partial_i s = \sum_{t \in \operatorname{supp}(s)} t$. In this context, a check, which is an operator, and a qubit, a vector in $\mathbb{C}^2$, are identified with the abstract cells representing them, denoted in lowercase such as $x \in X$ or $z \in Z$. Consequently, for a check $s$, $\operatorname{supp}(s)$ refers to the support of the row vector representing the check in the corresponding parity-check matrix, which can be identified with $\operatorname{supp}(\partial s)$, the support of the chain $\partial s$. These cells will also represent certain geometrical objects involved in the Freedman-Hastings handlebody realization of a code. In this setup, the boundary of a cell has a geometrical interpretation.\par

\subsection{Tanner-lift, cellular-lift and $\mathbb Z$-lift}\label{section Tanner-lift, cellular-lift and Z-lift}

Throughout this note, we use the term "lift" to refer to several distinct mathematical procedures. In this section, we provide a high-level overview of the main types of lifts considered: the Tanner lift, the FH lift, and the $\mathbb{Z}$-lift. Each of these constructions shares a common feature: starting from a length-3 $\mathbb{F}_2$-chain complex given with a basis, the lift replaces the field coefficients with elements from a ring or with matrix representations of group algebra elements, while preserving the orthogonality of the boundary maps.\par

We begin by recalling the definition of code lifting introduced in~\cite{GuemardLiftIEEE}, which we refer to as the \textit{Tanner lift}. For the definition of the Tanner cone-complex associated with a CSS code, see~\cite[Section~3.2]{GuemardLiftIEEE}.

\begin{definition}[Tanner-lift of a CSS code, \cite{GuemardLiftIEEE}]\label{definition lift of quantum CSS code}
Let $C$ be a CSS code, with Tanner graph $T:=\mathcal{T}(C)$, and Tanner cone-complex $K:=\mathcal{K}(C)$. Let $p:K'\rightarrow K $ be a finite cover and $p|_T:T'\to T$ its restriction to the Tanner graph. The \textit{Tanner-lift} of $C$ associated to $p$ corresponds to the CSS code $C'$ such that $\mathcal{K}(C')=K'$ (or equivalently with Tanner graph $T'=\mathcal{T}(C')$, sets of checks $X'=p^{-1}(X)$, $Z'=p^{-1}(Z)$ and qubits $Q'=p^{-1}(Q)$. A \textit{connected} lift of $C$ is a lift defined by a connected cover of $K$. A \textit{trivial lift }of $C$ corresponds to disjoint copies of this code.
\end{definition}
This technique has proven useful for constructing LDPC codes. Notably, the maximum check weight of $C'$ remains equal to that of $C$, a consequence of the properties of covering maps. Moreover, this framework naturally generalizes the classical code-lifting technique. Indeed, when $C$ is a linear code, the definition coincides with the standard lifting procedure, which corresponds to taking a covering of its Tanner graph.\par

For any code $C$, the Freedman-Hastings procedure \cite{Freedman2020CSS_Manifold} constructs a cellulated 11-dimensional manifold which includes the code in its cellular chain complex. Freedman and Hastings added a step in the construction to make it simply-connected. Skipping this step is crucial for code lifting since connected coverings of a manifold are in correspondence with subgroup of its fundamental group, by Theorem \ref{Theorem Galois correspond}. Because the construction is not unique, we denote the set of all possible handlebody realizations of $C$ by $\mathcal{M}(C)$, and a fixed choice of manifold realization by $M\in \mathcal M(C)$ . Then $C$ is the length-3 subcomplex of middle dimension 4 of the $\mathbb Z_2$-cellular chain complex of $M$. We will detail the construction in Section \ref{section the freedman hastings construction}.

\begin{definition}[cellular-lift]
Let $C$ be a CSS code, and let $M\in\mathcal{M}(C)$ be a fixed handlebody realization of $C$. Let $p:M'\rightarrow M $ be a finite cover. The \textit{cellular-lift }of $C$ associated to $p$ corresponds to the CSS code $C'$, such that $C'$ is the length-3 subcomplex of middle dimension 4 of the $\mathbb Z_2$-cellular chain complex of $M'$, i.e $M'\in \mathcal{M}(C')$.
\end{definition}
In Section~\ref{section Cellular realization of a code}, instead of considering a manifold realization, we examine a simplified version: a cellular complex of dimension 5 that realizes the code in dimensions 3, 4, and 5. Constructing a covering over this simplified cell complex gives a definition of the lift equivalent to the one given here. We will therefore identify the two notions.\par

A key distinction between the Tanner-lift and the cellular-lift is that the latter necessitates a preliminary step: constructing a \textit{$\mathbb{Z}$-lift} of the code.

\begin{definition}
Given a $\mathbb{F}_2$-chain complex $C$ endowed with a basis, and boundary maps $(\partial_n)_{n\in \mathbb{N}}$ interpreted as finite matrices, a $\mathbb{Z}$-lift of $C$ is a chain complex $\widetilde{C}$ with integer coefficients, defined by a sequence of boundary operators $(\widetilde{\partial}_n)_{n\in \mathbb{N}}$, such that, for all $n$, $\widetilde{\partial}_n = \partial_n \: (\operatorname{mod} 2)$.
\end{definition}
It was shown in~\cite{Freedman2020CSS_Manifold} that any finite-dimensional 3-term complex over $\mathbb{F}_2$, representing the code $ C := \mathbb{F}_2^{|Z|} \xrightarrow{H_Z^T} \mathbb{F}_2^{|Q|} \xrightarrow{H_X} \mathbb{F}_2^{|X|}$, can be $\mathbb{Z}$-lifted into another chain complex $\widetilde{C} := \mathbb{Z}^{|Z|} \xrightarrow{\tilde{H}_Z^T} \mathbb{Z}^{|Q|} \xrightarrow{\tilde{H}_X} \mathbb{Z}^{|X|}$. In general, a $\mathbb{Z}$-lift is not unique and can be given extra constraints, such as being torsion-free conditions or sparse\footnote{the sum of the absolute values of the coefficients over rows and columns is $O(1)$ in the number of qubits, when considering families of sparse chain complexes}. Moreover, we will be interested in \textit{support-preserving $\mathbb{Z}$-lifts} in Section \ref{section Relation between the cellular-lift and the Tanner-lift}.

\begin{definition}
    We call a $\mathbb Z$-lift of a CSS code support-preserving if it modifies only the non-zero coefficients of the parity-check matrices (into odd integers). 
\end{definition}

It is essential to note that certain chain complexes do not admit a support-preserving $\mathbb{Z}$-lift, as shown in Appendix~\ref{section A CSS code admitting no support preserving Z-lift}. We will come back to this issue in Section \ref{section Relation between the cellular-lift and the Tanner-lift} \par

\section{The Freedman-Hastings construction}\label{section the freedman hastings construction}

\subsection{Handlebody realization of a CSS code}\label{section Handlebody realization of a CSS code}

In~\cite{Freedman2020CSS_Manifold}, Freedman and Hastings proposed a method to construct a simply-connected, closed, oriented Riemannian 11-dimensional manifold inheriting its $\mathbb{Z}_2$-systolic properties from the $X$ and $Z$ relative distances of a quantum CSS code $C$. Here, we start by reviewing their work. In this section and the next, $n$, $k$, and $d$ refer to the dimensions of geometrical objects, and not to the parameters of a code.\par

Any 3-term complex over $\mathbb{F}_2$, here representing a CSS code $C := \mathbb{F}_2^{|Z|} \xrightarrow{\partial_2} \mathbb{F}_2^{|Q|} \xrightarrow{\partial_1} \mathbb{F}_2^{|X|}$, can be $\mathbb{Z}$-lifted into another complex $\widetilde{C}$. The principle of~\cite{Freedman2020CSS_Manifold} is to use $\widetilde{C}$ as an incidence complex, around which a handle decomposition of the manifold is built,
\[ \text{span}_\mathbb{Z}( 11\text{-handles}) \xrightarrow{\text{attaching}}  \text{span}_\mathbb{Z}( 10 \text{-handles})  \xrightarrow{\text{attaching}} \dots \xrightarrow{\text{attaching}}
 \text{span}_\mathbb{Z}( 0 \text{-handles}).\]
Such a complex is referred to as a handle complex. A $\mathbb{Z}$-lift is required because the information contained in an $\mathbb{F}_2$-chain complex is not enough to define attaching maps or ensure orientability of the manifold.\par

In the handle complex, the qubits correspond to the 4-handles, while the $X$ and $Z$-checks correspond to the 3 and 5-handles, respectively. The lifted code $\widetilde{ C }$ is then isomorphic to the 3-term complexes
\[ \text{span}_\mathbb{Z}( (d+5) \text{-handles}) \xrightarrow{\text{attaching}}  \text{span}_\mathbb{Z}( (d+4) \text{-handles})  \xrightarrow{\text{attaching}}
 \text{span}_\mathbb{Z}( (d+3) \text{-handles}), \]
for $d=0$ or $3$, by duality.\par

A more familiar correspondence between the original chain complex of the code, and the $k$-handle structure of the manifold is obtained through deformation retract of $k$-handles into $k$-cells; we will explain this aspect in Section \ref{section Cellular realization of a code}.\par

The dimension of the manifold, and the indices of the handles, have been chosen high enough to ensure a separation between the handles coming from the $\mathbb{Z}$-lifted code, of index 3, 4, 5, and dually 6, 7, 8, and the handles of index 0, 1, 2, and dually 9, 10, 11, whose purpose is to organize the attachments and give the right incidence numbers between the handles of index 3, 4, 5, and 6, 7, 8. In particular, it avoids the appearance of extra unwanted homology in dimensions 4 and 7.\par

When a code family $\{C_i\}_\mathbb{N}$ is LDPC and admits sparse $\mathbb{Z}$-lifts $\{\widetilde{ C }_i\}_\mathbb{N}$, it guarantees the existence of an associated family of manifolds $\{M_i\}_\mathbb{N}$ with a bounded Riemannian metric. While this is a key aspect in \cite{Freedman2020CSS_Manifold}, for our purposes, it is not necessary for the codes to be LDPC. We will explain why in Section \ref{section cellular-lift of a CSS code}.\par

We recall some standard definitions and summarize the construction from \cite{Freedman2020CSS_Manifold}. An $n$-dimensional handlebody is a decomposition of a manifold into a union of $n$-dimensional handles, with indices ranging from $0$ to $n$. A $n$-dimensional $k$-handle can be described as a pair
\begin{equation*}
    h^k=\{D^k\times D^{n-k},\partial D^k\times D^{n-k}\},
\end{equation*}
where the second component of the boundary $\partial h^k = S^{k-1} \times D^{n-k} \cup D^k \times S^{n-k-1}$ is called the co-attaching region. This region remains available after attachment, allowing the gluing of handles of higher indices. \par

The construction of a handlebody begins with 0-handles, $(D^0 \times D^n, \varnothing)$, which have an empty attaching region. Handles of index $k$ are then attached along their attaching region to the co-attaching region of handles of index $k-1$. In \cite{Freedman2020CSS_Manifold}, the final step involves attaching 11-handles, which have an empty co-attaching region, to close the 11-dimensional manifold. \par

While the standard approach is to attach handles in order of increasing index, this is not strictly necessary. Handles can always be reordered by index using handle-slides.

A handle decomposition can be turned upside-down. A handle of index $k$ becomes a dual handle, i.e., one of index $n-k$. The resulting handlebody is isomorphic to the previous one, and the local geometry is preserved. This result follows from arguments on Morse functions. Moreover, contact between handles is a symmetric relation, and one obtains a dual handle-complex.\par

The boundary maps of the handle-complex are given by the incidence numbers of pairs of handles $h^k$ and $h^{k-1}$. These incidence numbers correspond to the intersection number between the \textit{attaching sphere} of $h^k$ and the \textit{belt sphere} of $h^{k-1}$. For a handle $h^k$, the attaching sphere is the boundary $S^{k-1} \times 0$ of its \textit{core} $D^k \times 0$, and the belt sphere is the boundary $0 \times S^{n-k-1}$ of its \textit{co-core} $0 \times D^{n-k}$, where $0$ denotes the center of the disks $D^k$ or $D^{n-k}$.\par

A cell decomposition is obtained by replacing each handle with an equivalent cell, which is done by retracting each handle onto its core and attaching region. Therefore, we have that $\widetilde{C}_k = \mathbb{Z}^{\#_k}$, where $\#_k$ denotes the number of attached $k$-handles, and the incidence numbers are preserved.
\par

\begin{remark}
Although for our purposes it will be sufficient to construct a cell complex, it is easier to understand the logic of the construction through the description of the handlebody. Indeed, when retracting handles to their core, many handle attachments become identified, making it a tedious task to recognize the generators of the homology groups.
\end{remark}

To explain the construction of Freedman and Hastings, we now restrict ourselves to $n=11$.\footnote{We do not enter into the details of how to fix the diffeomorphism type of the manifold. Instead, we refer the interested reader to \cite{Freedman2020CSS_Manifold} for a complete treatment}. We recall that $X$ and $Z$ denote the set of $X$ and $Z$ stabilizer generators, respectively, that we refer to as stabilizers, and $Q$ is the set of physical qubits. The construction is divided into three main steps. First, a 5-handlebody, i.e. a union of handles with indices ranging from 0 to 5, denoted $M_{ZQX}$ is built\footnote{This manifold is simply denoted $ZQX$ in \cite{Freedman2020CSS_Manifold}.}. The fundamental group $\pi_1(M_{ZQX})$ is freely generated by the 1-handles. It is trivialized by adding 2-handles, resulting in a manifold denoted $M_{ZQX^+}$ . Over concentration of 2-handles is avoided by using a certain weakly fundamental cycle basis that the author design appropriatly. In Section \ref{section cellular-lift of a CSS code}, we will ignore this step and work directly with $M_{ZQX}$.\par

Then, its upside-down copy $M_{ZQX^-}$, with dual ($11-k$)-handles and reversed orientation, is glued to $M_{ZQX^+}$ along their identical boundary, $\pm\partial M_{ZQX}$. This forms the closed oriented manifold,
\[M^{11}=M_{ZQX^+}\cup_\partial M_{ZQX^-},\]
which has handles of indices ranging from 1 to 11. Therefore, we only have to detail the construction of $M_{ZQX}$ and $M_{ZQX^+}$ to specify $M^{11}$. \par

To construct $M_{ZQX}$ we will not attach handles in the standard order of increasing index. Instead, handles of indix 0,1 and 2 will be introduced at various steps of the construction. The 3-handles ($X$-stabilizers) will be attached in a straightforward manner, but this will be more laborious for 4-handles (qubits) and 5-handles ($Z$-stabilizers). For them, we will rather create \textit{dressed-handles} $h_q$ and $h_z$, corresponding respectively to 4-dressed-handles and 5-dressed-handles. These are not standard handles, but each of them is a handle-body possibly involving 1 and 2-handles. However, each $h_q$ or $h_z$ contains, a single 4 and 5-handle respectively. Thus, each time we attach a 4 or a 5-dressed-handle, we effectively attach a single 4 or 5-handle corresponding to a basis vector of the $\mathbb Z$-lifted complex $\widetilde{ C }$. Moreover, dressed-handles will be glued by order of increasing index. While not strictly necessary, but for clarity, we will also create a dressed-handle $h_x$ for each stabilizer $x\in X$.\par

We will denote $h_x$, $h_q$ and $h_z$ the dressed-handle realization of $x\in X$, $q\in Q$ and $z\in Z $. They will be first presented in the form 
\begin{align*}
    &h_x=\{c_x \times D^8 ,\partial c_x\times D^8 \}, \\
    &h_q=\{c_q \times D^7 ,\partial c_q\times D^7 \}, \\
    &h_z=\{c_z\times D^6 ,\partial c_z \times D^6 \},
\end{align*}
and later we will detail their internal structure. The component $c_x$, $c_q$, $c_z$ called their \textit{dressed-core} and are also handlebodies of lower dimension. Both the dressed-handles $h_x$, $h_q$, $h_z$ and the dressed-cores $c_x$, $c_q$, $c_z$ can be seen as basis vectors of $\widetilde{C}_{3},\widetilde{C}_{4}$ or $\widetilde{C}_{5}$, respectively.\par
$M_{QX}$ is the manifold obtained by gluing the set of 4-dressed-handles $\{ h_q\: : \: q\in Q\}$ to the set of 3-dressed-handles $\{ h_x\: : \: x\in X\}$ according to the incidence numbers specified by the bundary map $\tilde{\partial}_1$. Similarly, $M_{ZQX}$ is the manifold obtained by attaching the set of 5-dressed-handles $\{h_z\:: \: z\in Z\}$ to $M_{QX}$.\par 

Then, the isomorphism between the handle-complex and the $\mathbb Z$-lifted code guarantees isomorphism between homolgy groups
\[H_k(M_{ZQX},\mathbb{Z}) \cong H_k(M_{ZQX^\pm},\mathbb{Z}) \cong H_k(M^{11},\mathbb{Z}) \cong H_{k-3}(\widetilde{C},\mathbb{Z}),\]
for $k=3,4$ and $5$, and for the dual dimensions. The corresponding chain complex being isomorphic, this also holds for their $\mathbb{Z}_2$ reduction.\par

Finally, Freedman and Hastings prove that the $\mathbb{Z}_2$ $k$-systoles of $M^{11}$, defined as 
\[sys_k(M^{11})=\min \limits_{0\neq[c]\in H_k(M^{11},\mathbb Z_2)}vol(c), \] and in particular the ratios $sys_k(M^{11})/vol(M)$ for $k=4$ and $k=7$, are governed by the $X$ and $Z$ relative distance of the code $C$, respectively. Using as input either the fiber bundle code \cite{Hastings2020} or the asymptotically good quantum LDPC code family of \cite{Panteleev2021}, they give the first known families of (smooth, closed, orientable) manifolds exhibiting  $\mathbb{Z}_2$-systolic freedom.

\subsection{Structure of $M_{QX}$ and 4-handles}\label{section Structure of M_QX and 4-handles}

In this section, we clarify the structure of the subcomponents of $M_{QX}$. An element $x\in X$ has a geometric realization which is a dressed-handle
\[h_x=\{c_x\times D^8,\partial c_x\times D^8\}=\{S^3\times D^8,\varnothing\}.\]
To construct it, simply attach a 3-handle $h^3$, which corresponds to the basis vector of $x$ in the chain complex $\tilde C$, to a 0-handle. The result is diffeomorphic to $S^3\times D^8$. In total, there are $|X|$ of these objects and this results in a disconnected $(0,3)$-handlebody isomorphic to $\sqcup_{x\in X} S^3 \times D^8  $ called $M_X$. This manifold also represents the trivial chain complex $\tilde{C}_0\rightarrow 0$.

\begin{figure}[t]
  \centering
  \includegraphics[scale=1]{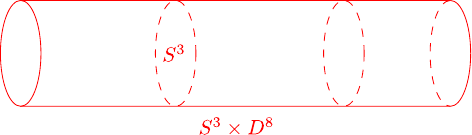}
\caption{A dimensionally reduced representation of a dressed 3-handle as a product bundle of a 3-sphere and an 8-disk. Sections of this bundle, represented as dashed spheres, are where the 4-dressed-handles can be attached (along their punctures).}
  \label{fig:3-handle}
\end{figure}
A 4-\emph{handle} associated to a qubit $ q \in Q $ is a pair
\[
h_q = \left\{ c_q \times D^7,\, \partial c_q \times D^7 \right\} = \left\{ \left(S^4 \setminus \bigsqcup_t D^4 \right) \times D^7,\, \left( \bigsqcup_t S^3 \right) \times D^7 \right\}.
\]
The manifold factor $ c_q $ is a punctured sphere,
\[
c_q = S^4 \setminus \bigsqcup_t D^4,
\]
where $t = \sum_{x \in X} \left| \left( \widetilde{\partial}_1 \right)_{x,q} \right|$, and we use the notation $\left( \widetilde{\partial}_1 \right)_{x,q}$ to denote the incidence number between $ x $ and $ q $ in the complex $ \widetilde{C} $. Each boundary component of $ c_q $ is to be attached to the corresponding $ x $'s according to $ \widetilde{\partial}_1 $, via disjoint embeddings
\[
S^3 \times D^7 \hookrightarrow S^3 \times S^7
\]
with homological degree $ \pm 1 $ in $ H_3(\partial h_x, \mathbb{Z}) $. Summing over these degrees give the corresponding entry in $ \widetilde{\partial}_1 $. We give an example procedure to construct $ h_q $.

\begin{figure}[t]
  \centering
  \includegraphics[scale=1]{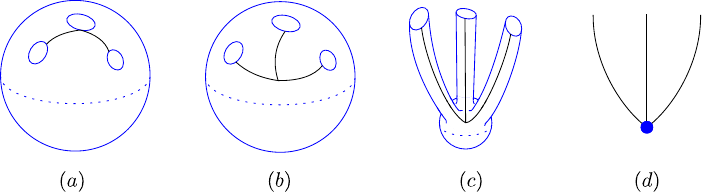}
  \caption{(a) Handle structure of a 4-dressed-handle. 1-handles are represented as segments, and 2-handles as disks. This corresponds to the cell structure when retracting handles to their core. (b) Modified cell structure discussed in Section \ref{section Cellular realization of a code}, with the addition of a 0-cell and a 1-cell. (c) Homeomorphic transformation of the dressed-handle, resulting in (d), a graph representation of the dressed-handle, with a blue vertex representing the 4-cell and half-edges corresponding to the punctures.}

  \label{fig:puncturedsphere}
\end{figure}

\begin{example}
Relative to its boundary, $h_q$ can be constructed with $(t-1)$ 1-handles, to connect the disjoint boundaries, and one 4-handle. However, it is more informative to construct its core $c_q$. The handle structure of $c_q$ is obtained via $(t-1)$ 1-handles of dimension 4, and one 4-handle. It is instructive to obtain the same object by turning the construction upside-down, namely with one 0-handle and $t$ 3-handles $h^3 = \{D^3 \times D^1, S^2 \times D^1\}$. We can easily see inductively that $c_q = \natural_{t-1} S^3 \times D^1$, where $\natural$ denotes the boundary connected sum. Indeed, after attachment of a 3-handle to the 0-handle, the result is isomorphic to $S^3 \times D^1$ and has boundary $S^3 \times S^0$. Attaching another 3-handle is equivalent to starting with two copies of $S^3 \times D^1$ and identifying a pair of 3-disks on their boundary, or equivalently connecting a 1-handle between the two. This operation is the boundary connected sum. We also verify that $\partial \natural_{t-1} S^3 \times D^1 = \#_{t-1} S^3 \times S^0 = \sqcup_t S^3$.

\end{example}

The manifold $M_{QX}$ is obtained by creating all dressed-handles corresponding to elements of $X$ and $Q$, and gluing them together according to the incidence number specified by $\tilde{\partial}_1$.

\subsection{Generators of the fourth homology groups}\label{section Generators of the fourth homology groups}

To explain the construction of the 5-dressed-handles (corresponding to $Z$-checks), it is helpful to first understand how they are attached to the boundary of $M_{QX}$. We therefore start by clarifying the nature of the generators of $H_4(\partial M_{QX})$, which correspond to the regions where 5-dressed-handles need to be attached. These generators are closed manifolds, denoted $Y$ or $Y_i$ here and in the next section, each isomorphic to $\#_t S^3 \times S^1$ for some $t \in \mathbb{N}$. To build intuition, we first present a simplified model in the example below, which we then generalize.

\begin{example}\label{example generators of fourth homology}
We consider a closed 4-manifold $Y$ obtained by combining the core of a qubit $c_q = S^4 \setminus \sqcup_{2t} D^4$ and $t$ $X$-checks, all of the form $c_x \times D^1$ (note that here we have started with $h_x$ and contracted 7 dimensions of the disk factor $D^8$), such that $c_q$ is attached twice to every check. The handle decomposition is obtained by counting the number of 1-handles needed to connect all the disconnected boundaries of $\sqcup_x c_x \times D^1$ before gluing the 4-handle representing $q$. Therefore, we have: $t$ 0-handles, $(2t - 1)$ 1-handles, $t$ 3-handles, and one 4-handle.\par

We now wish to give the manifold $Y$ a more familiar description, namely that $Y = \#_t S^3 \times S^1$. For this, we turn the decomposition upside down and obtain one 0-handle, $t$ (untwisted) 1-handles, $(2t - 1)$ 3-handles, and $t$ 4-handles. We call $W_-$ the open manifold composed of the 0-handle and 1-handles, and $W_+$ the one with the higher handles. The decomposition of $W_-$ is standard: $W_- = \natural_t S^1 \times D^3$, and since $Y$ is closed, they satisfy $\partial W_- = -\partial W_+ = \#_t S^1 \times S^2$. By a theorem of Laudenbach and Poénaru~\cite{BSMF_1972__100__337_0}, any diffeomorphism of $\#_t S^1 \times S^2$ extends to one of $\natural_t S^1 \times D^3$. Therefore, $W_+ = \natural_t S^1 \times D^3$. By gluing the lower and upper parts, we obtain $Y = \#_t S^3 \times S^1$.\par

We note that $S^3 \times S^1$ can be obtained from a twice-punctured 4-sphere by identifying its boundaries. The same manifold $Y$ can therefore be obtained from a $2t$-times punctured 4-sphere by identifying pairs of boundaries. This is equivalent to what we have done previously.

\end{example}

\begin{example}\label{example 2 generators of fourth homology}
We extend the model described in Example \ref{example generators of fourth homology} to a closed $4$-manifold $Y$ consisting of $r$ qubits $\{c_{q_i}\}$ and $t$ $X$-checks $\{c_{x_i} \times D^1\}$. For simplicity, we assume that the total number of punctures is $2t$, and that each check is connected to two punctures. Gluing two disconnected pairs of $t_1$ and $t_2$-punctured spheres, $c_{q_1}$ and $c_{q_2}$, to a core $c_x$ results in a $(t_1 + t_2 - 2)$-punctured sphere. Thus, we first consider $r$ checks to connect all the punctured spheres together, resulting in one $2(t - r)$-punctured sphere. Then, using the argument from Example \ref{example generators of fourth homology}, we attach the remaining $X$-checks to obtain $Y = \#_{(t - r)} S^3 \times S^1$.
\end{example}

The manifold described in Example \ref{example 2 generators of fourth homology} is typical of embedded $4$-manifolds in $\partial M_{QX}$. We can leverage this idea to attach the 5-handles with the correct incidence numbers. Multiple disjoint copies of a manifold $c_q$ can be found as sections of $h_q$, and similarly, multiple disjoint copies of a manifold $c_x$ can be found as sections of $h_x$. Therefore, it is possible to consider an embedded $4$-manifold constructed from any set of copies $\{c^{(i)}_{q}\}$ and $\{c^{(i)}_{x}\}$, where the elements are chosen as arbitrary sections of $\{h_q\}_{q\in Q}$ and $\{h_x\}_{x\in X}$, with the only condition being that a section $c^{(i)}_{q}$ of $h_q$ can be connected to a section $c^{(j)}_{x}$ of $h_x$ only if $h_q$ and $h_x$ are connected in $M_{QX}$, and as long as the overall pairing of $\{c^{(i)}_{q}\}$ and $\{c^{(i)}_{x}\}$ respects the construction presented in Example \ref{example 2 generators of fourth homology}.\par

These embedded orientable closed $4$-manifolds generate $H_4(\partial M_{QX})$ and serve as attachments to the 5-dressed-handles $\{h_z\::\: z \in Z\}$, transforming certain elements of $H_4(M_{QX})$ into boundaries\footnote{Note that, by a theorem of Thom, the generators of homology groups of a manifold are not always embedded sub-manifolds.}. After attachment, these $4$-manifolds also serve as representatives of elements in $H_4(M_{ZQX})$.

\subsection{Structure of $M_{ZQX}$ and 5-handles} \label{section Structure of M_ZQX and 5-handles}

A stabilizer $z \in Z$ has a dressed-handle realization written as a pair $h_z = \{c_z \times D^6, \partial c_z \times D^6\}$, where $c_z$ is a 5-dimensional manifold with a boundary that must be attached to a representative in $H_4(\partial M_{QX})$. The manifold $c_z$ can be built as a 5-dimensional sphere $S^5$ to which we remove a thickened graph. This graph, denoted $G_z$, is designed to give the correct incidence between $h_z$ and the 4-handles of the qubits. We first outline a series of steps to build $G_z$ and explain how to construct $c_z$ and $h_z$.\par

Let $\widetilde{T}$ represent the $\mathbb{Z}$-lifted Tanner graph of $\mathbb{Z}^{|Q|}\xrightarrow{\tilde{\partial}_1 }\mathbb{Z}^{|X|}$. It is similar to the original Tanner graph $T$, but an edge in $T$ is replaced by a number of multi-edges in $\tilde{T}$, where the number of these multi-edges is given by the corresponding entry in $\tilde{\partial}_1$. Each multi-edge is assigned a sign $\pm$, denoted $r_{q,x}$, where the sign is given by the homological degree of the embedding of $\partial h_q$ into $\partial h_x$, or equivalently $r_{q,x} = \mathrm{sgn}(\tilde{\partial}_1)_{q,x}$. We can also think of this graph as the signed bipartite multi-graph \cite{zaslavsky2013matricestheorysignedsimple} with adjacency matrix
\[
\renewcommand\arraystretch{1.3}
A=\mleft[
\begin{array}{c|c}
  0 & \tilde{\partial}_1  \\
  \hline
  \tilde{\partial}_1^T & 0 \\
\end{array}
\mright].
\]
We call the subgraph induced by a generator $z \in Z$ in $\widetilde{T}$ the (possibly disconnected) subgraph $\widetilde{T}_z$ of $\widetilde{T}$, composed of every $Q$-vertex in the support of $z$, every $X$-vertex having these qubits in their support, and all the signed edges between them.\par

Let $\widetilde{T}^{\text{mult}}_z$ denote the graph built from $\widetilde{T}_z$, where each $Q$-vertex appears in a number of copies given by $|(\tilde{\partial})_{z,q}|$, and with signed edges, where an edge between $q$ and $x$ is assigned the sign $r_{z,q} \cdot r_{q,x}$, with $r_{z,q} = \mathrm{sgn}(\tilde{\partial}_2)_{z,q}$. Notice that at each $X$-vertex, the number of positively signed edges equals the number of negatively signed edges. This follows from the chain complex property $\tilde{\partial}_1 \circ \tilde{\partial}_2 = 0$. We can therefore choose a pairing of oppositely signed edges at every $X$-vertex of $\widetilde{T}^{\text{mult}}_z$ and glue each pair to a new copy of the corresponding $X$-vertex. As a result, we obtain several copies of each vertex $x$, all of degree 2. The resulting graph, possibly disconnected\footnote{The graph can be disconnected if $\widetilde{T}_z$ is disconnected or due to the choice of edge pairing.}, is the graph introduced at the beginning of this section, denoted $G_z = \sqcup_i G_{z,i}$. It is realized as an embedded graph in the boundary of $M_{QX}$, and we can use it to define an embedded submanifold $Y_z = \sqcup_i Y_{z,i}$, with
\[
Y_{z,i} = \#_{b_1} S^3 \times S^1,
\]
where $b_{i,1}$ is the first Betti number of $G_{z,i}$. As illustrated in Section \ref{section Generators of the fourth homology groups}, the manifold $Y_{z,i}$ is a typical representative of an element in $H_4(\partial M_{QX})$.

The 5-dressed-handle $h_z$ must be glued to $Y_z$ with homological degree $\pm 1$. For this purpose, we can give $h_z$ a simple description,
\[
h_z = \{(S^5 \setminus \text{int} N(G_z)) \times D^6, \partial(S^5 \setminus \text{int} N(G_z)) \times D^6\},
\]
where $N(G_z)$ denotes a closed regular neighborhood of $G_z \hookrightarrow S^5$, and $\mathrm{int}$ denotes its interior.\par

\begin{figure}[t]
  \centering
  \includegraphics[scale=0.9]{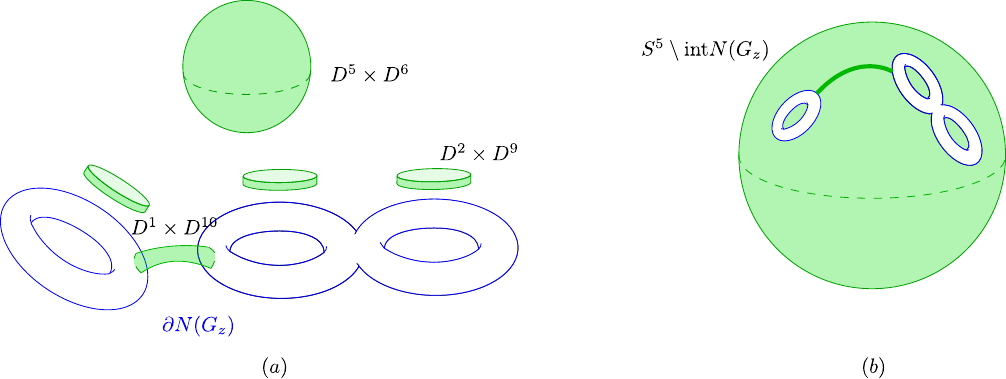}
  \caption{(a) Example of a handle structure of a 5-dressed-handle, where its components are represented in green, and the 5-handle is represented by a 3-disk. The blue region is a boundary component of the manifold $M_{QX}$. It corresponds to a the boundary of a thickened graph $G_z$, with 2 connected components. The construction of the 5-dressed-handle consists in joining them with a 1-handle, and triviliazing the fundamental group by adding 2-handles. The 5-handle is attached at the end along the resulting boundary. (b) Dimensionally reduced representation of the resulting manifold, where the 1-handle is represented by a thick green line, and the 5-handle by a 2-disk. }
  \label{fig: 5-handle}
\end{figure}
Then we can build a handlebody structure for $h_z$ relative to its boundary. Use 5-dimensional 1-handles to connect all the boundary components of $c_z = S^5 \setminus \text{int} N(G_z)$; recall that $G_z$ can be disconnected. Kill the fundamental group of the new attaching region with 2-handles by locating just enough generators. The handlebody structure of $c_z$ is completed by capping off with a single 5-handle. Crossing the result with $D^6$, each handle preserves its index, and we obtain a handle decomposition of the dressed 5-handle $h_z$. Repeating this process for every element of $Z$ completes the construction of $\{h_z : z \in Z\}$ and $M_{ZQX}$.\par
In Figure \ref{fig: 5-handle}, we illustrate an example of a handle structure for a 5-dressed-handle built from a graph $G_z$ with two connected components. Notice that the 1-handle connecting these components represents a generator of the fundamental group of $M_{ZQX}$.\par

\begin{example}\label{example 5-handle to MQX}
Consider the CSS code defined by the parity-check matrices
\[
H_X = \mleft[
\begin{array}{cc}
  1 & 1  \\
  1 & 1 \\
\end{array}
\mright], \qquad H_Z = \mleft[
\begin{array}{cc}
  1 & 1  \\
\end{array}
\mright].
\]
This code can be represented by the Tanner graph in Figure~\ref{fig:Tanner graphs}~(a). To make this example more illustrative, suppose we choose the following $\mathbb{Z}$-lift of these matrices:
\[
\tilde{H}_X = \mleft[
\begin{array}{cc}
  -3 & 1  \\
  -3 & 1 \\
\end{array}
\mright], \qquad \tilde{H}_Z = \mleft[
\begin{array}{cc}
  1 & 3  \\
\end{array}
\mright].
\]
This represents a valid $\mathbb{Z}$-lift of the chain complex, although not the simplest. The manifold $M_{QX}$ is depicted in Figure~\ref{embedded manifold_diagram}, using the representation of the 3- and 4-dressed-handles established in previous sections. According to the $\mathbb{Z}$-lift above, the 5-handle, corresponding to the single row of $\tilde{H}_Z$, should have an attaching map of degree $+1$ on the first qubit (4-handle), and degree $+3$ on the second qubit.\par

We now describe how to attach the 5-dressed-handle $h_z$ with respect to these degrees. The signed multigraph induced by the single generator $z$ is shown in Figure~\ref{fig:Tanner graphs}~(b), and the graph $\widetilde{T}_z^{\text{mult}}$ in Figure~\ref{fig:Tanner graphs}~(c). To enable a pairing of positively and negatively signed edges at the $X$-vertices, we create three copies of the second qubit—corresponding to the second entry of $\tilde{H}_Z$. An example of such a pairing is shown in Figure~\ref{fig:Tanner graphs}~(d), resulting in the graph $G_z$. This graph corresponds to a manifold isomorphic to the connected sum $\#_3(S^1 \times S^3)$, as illustrated in Figure~\ref{fig:Tanner graphs}. This manifold serves as attaching region for $h_z$ in the manifold $M_{QX}$ and is represented via the embedding of $G_z$ in Figure~\ref{embedded manifold_diagram}. Notice that the manifold goes into the 4-handle of the first qubit once and the 4-handle of the second qubit three times. This is why the attaching map of $h_z$ has the correct homological degree.

\end{example}

\begin{figure}[t]
  \centering
  \includegraphics[scale=0.9]{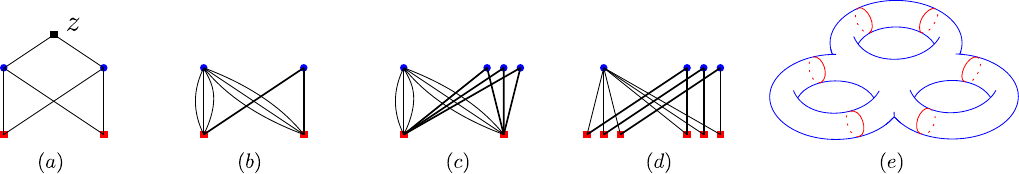}
  \caption{(a) Tanner graph of the code considered in Example~\ref{example 5-handle to MQX}, with qubits shown in blue and $X$-checks in red. (b) Signed multigraph $\widetilde{T}$ of the code $C_X$, which also corresponds to the induced subgraph $\widetilde{T}_z$ in this example. Positively signed edges are depicted as thick lines, and negatively signed edges as thin lines. (c) The corresponding $\widetilde{T}_z^{\text{mult}}$, where each qubit appears in a number of copies determined by the corresponding entry of $H_Z$. (d) Pairing of positively and negatively signed edges at the $X$-vertices, resulting in the graph $G_z$ with Betti number $b_1 = 3$. This graph represents the manifold $Y_z$. (e) Manifold $Y_z = \#_3 (S^3 \times S^1)$ corresponding to the attaching region of the 5-dressed-handle $h_z$. The four blue components correspond to the qubit copies in $G_z$ (punctured spheres), while the six red spheres represent the $X$-vertices (3-handles).}
  \label{fig:Tanner graphs}
\end{figure}

\begin{figure}[t]
  \centering  \includegraphics[scale=1]{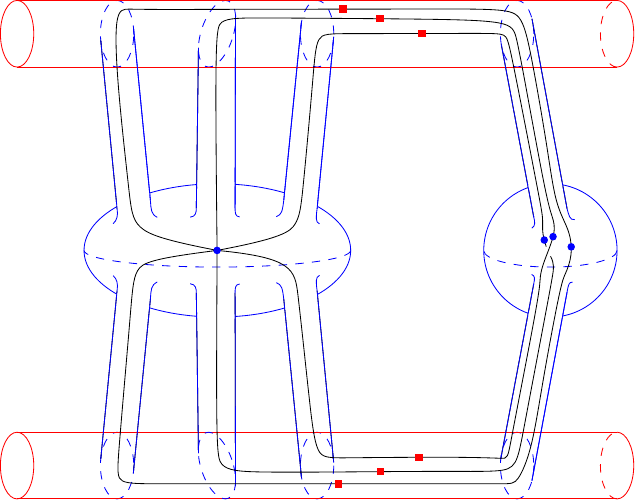}
\caption{Manifold $M_{QX}$ in Example~\ref{example 5-handle to MQX}, where the two 3-dressed-handles ($X$-checks) are shown in red, and the two 4-dressed-handles (qubits) in blue. This representation should be interpreted as a $4$-manifold crossed with a $7$-disk. Copies of manifolds of the form $Y_z$ are embedded within it, here represented by the graph $G_z$ (see Figure~\ref{fig:Tanner graphs}). This embedded manifold serves as a valid attaching region for the 5-dressed-handle $h_z$.}
  \label{embedded manifold_diagram}
\end{figure}

\section{Cellular-lift of a CSS code}\label{section cellular-lift of a CSS code}

\subsection{Cellular realization of a code}\label{section Cellular realization of a code}

We will now relax certain aspects of the construction of $M^{11}$ that are of minor importance for the purpose of performing code lifting: we only consider the handle structure of $M^{11}$ up to handles of index $5$; we also omit the step of killing its fundamental group by adding $2$-handles. This corresponds to the manifold\footnote{Note that this manifold is possibly simply connected.} with boundary $M_{ZQX}$. We will retract all handles to their core and consider a slightly modified associated cell complex, which is equivalent to relaxing the smoothness hypothesis. Moreover, the family of quantum CSS codes will not be required to be LDPC. The cell complex associated with each code has finitely many cells. \par

We also slightly modify the construction of the cell complex. We can refer back to the handle structure of $M_{ZQX}$ at each step to explain the changes. On every dressed core $c_q$, associated with qubit $q$, we introduce a new $0$-cell, as well as a new $1$-cell (see Figure~\ref{fig:puncturedsphere}). The set of vertices can thus be divided into $X$-vertices (the retracted $0$-handles) and $Q$-vertices (the new $0$-cells). Every $h_q$ has a "star" of $1$-cells attached to its $Q$-vertex, and each $2$-cell introduced for $Z$-attachment can be glued to qubits along simple cycles represented by a path of the form 
\[
f = [x_1, q_1].[q_1, x_2].[x_2, q_2].\dots.[q_k, x_1],
\]
where $[a,b]$ denotes an oriented edge from $a$ to $b$. This step can be carried out in such a way that the closure of each $2$-cell in the complex is simply-connected. We repeat the process for every element of $\{\widetilde{T}_z : z \in Z\}$. \par

Finally, we kill certain specific generators of the fundamental group of the resulting cell complex. Indeed, in Section~\ref{section Structure of M_ZQX and 5-handles}, a $5$-dressed-handle $h_z$ is built according to a graph $G_z$, which may have several connected components. This is illustrated in Figure~\ref{fig: 5-handle}. The $1$-handles connecting these components are generators of the fundamental group of the cell complex. In the cell complex, we refer to these 1-cells as $Z$-type, while the previous ones are $Q$-type 1-cells. Here, we choose to kill these generators by attaching $2$-cells. We assume that these added $2$-cells are part of the dressed $5$-handles. Our motivation is to ensure that each $5$-dressed-handle has a simply connected closure in the cell complex. \par

We denote by $\mathcal{L}(C)$ the set of all $5$-dimensional cell complexes constructed in this manner, where we always assume that the same choice of $\mathbb{Z}$-lift has been fixed for every element of $\mathcal{L}(C)$.

\begin{definition}
Let $C$ be a CSS code. Any cell complex $L\in \mathcal{L}(C)$ built as described above is referred to as a \textit{cellular realization} of $C$.
\end{definition}
In the following, given a cellular realization $L\in \mathcal L(C)$, we will often consider its $k$-skeleton denoted $L^k$.\par

The introduction of the new 1- and 0-cells suggests the diagrammatic notation for a punctured 4-sphere used in Section \ref{section Structure of M_QX and 4-handles}. It becomes clear that the graph (d) of Figure \ref{fig:puncturedsphere} is sufficient to represent a punctured sphere. We can therefore faithfully represent the embedded 4-manifolds to which 5-handles are attached directly by the graph $G_z$, as in Figure \ref{fig:Tanner graphs}. \par

Lastly, notice that the attachment of the 5-cell onto $M_{QX}$ defines a projection map $\rho: G_z \to \widetilde{T}_z$. An example of this mapping can be inferred from Figures~\ref{fig:Tanner graphs} and~\ref{embedded manifold_diagram}. It is crucial to observe that when we retract handles to their core (to obtain a cell decomposition), the gluing map of the 5-cell associated with each $z \in Z$ amounts to projecting $G_z$ onto $\widetilde{T}_z$ via $\rho$. Hence, if we are not interested in the handles but only in the associated cell complex, we must still use the graph $G_z$ to attach the 2-cells and the 5-cell. However, upon projecting $G_z$ back to $\widetilde{T}_z$, many of the 2-cell boundaries are mapped by $\rho$ onto overlapping cycles in $\widetilde{T}_z$, and some of them become identified.\par

\subsection{Definition of the cellular-lift}\label{section definition of the cellular-lift}

In this section, we leverage the cellular realization of a CSS code to define a natural notion of code lifting. We also describe the complete procedure to construct such a lift explicitly. The procedure relies on \cite{HatcherTopo} and \cite{Lyndon2001}, and is similar to the one used in \cite{GuemardLiftIEEE}. For any cell complex $L$, we denote its cellular chain complex with coefficients in a ring $R$ as $\mathcal{C}(L, R)$.

\begin{definition}\label{definition Geometrical lift of quantum CSS code}
Let $L\in \mathcal{L}(C) $ be a cellular realization of a quantum CSS code $ C $, as described in Section \ref{section Cellular realization of a code}, and let $p:L'\rightarrow L $ be a finite covering. The cellular-lift $C'$ of $ C $ associated to $p$ corresponds to the chain-complex
\[
 C':= \mathcal C_5( L' , \mathbb{Z}_2)\xrightarrow[]{\partial_5'}\mathcal C_4( L' , \mathbb{Z}_2)\xrightarrow{\partial_4'} \mathcal C_3( L' ,\mathbb{Z}_2).
\]
\end{definition}

Using the notation above, we can also say that $C'$ is a CSS code such that $L' \in \mathcal{L}(C')$, for a given choice of $\mathbb{Z}$-lift.

All cellular realizations of a code $C$ are connected finite CW complexes. They are therefore \textit{well-behaved} topological spaces, i.e., \textit{locally path-connected} and \textit{semi-locally simply connected}, which are necessary properties for the classification theorem known as the Galois correspondence, Theorem \ref{Theorem Galois correspond}, to apply.

\begin{theorem}[Galois correspondence]\label{Theorem Galois correspond}
Let $L$ be a well-behaved topological space. There is a bijection between the set of basepoint-preserving isomorphism classes of path-connected covering spaces $p:(L',v')\to (L,v)$ and the set of subgroups of its fundamental group $\pi_1 (L, v )$, obtained by associating the subgroup $p_*\pi_1 (L',v' )$ to the covering space $(L',v')$. \par 
Given a subgroup $H\leq \pi_1(L)$ the degree $d$ of the covering is given by the index\footnote{Here, for groups $H\leq G$, $[G:H]$ is the standard notation for the index of $H$ in $G$.} $d=[\pi_1(L):H]$.
\end{theorem}
Let $H$ be a subgroup of $\pi_1(L)$, and
\[
p_H : L_H \rightarrow L
\]
the corresponding covering. Our objective here is to construct the space $L_H$. For this, we will first generate a covering of $L^2$, and then show how to complete a covering of $L^2$ into one of $L$.

Recall that $L^1$ is the union of the $\mathbb{Z}$-lifted Tanner graph of the chain complex $\widetilde{C}_0 \xrightarrow{\widetilde{H}_X} \widetilde{C}_1$ with the $Z$-type 1-cells. We denote the set of vertices and edges by $(V, E)$. Let 
\[
\phi : \pi_1(L^1) \to \pi_1(L)
\]
be the homomorphism defined by taking the quotient of the free group $\pi_1(L^1)$ by the normal closure of the subgroup generated by the loops corresponding to the boundaries of the $2$-cells in $L$. We define $H^1 := \phi^{-1}(H)$, i.e., the preimage of $H$ in $\pi_1(L^1)$. Whenever $H$ is normal, there is an isomorphism $\pi_1(L)/H \cong \pi_1(L^1)/H^1$, otherwise, this is a bijection between the cosets\footnote{Here, all cosets are taken on the right.}.

To construct a covering map of the $2$-skeleton, we must start with one of the $1$-skeleton,
\[
p_{H^1} : L^{1}_{H^1} \rightarrow L^{1}.
\]
This procedure follows the approach outlined in \cite{GuemardLiftIEEE}, and is based on methods described in \cite{HatcherTopo} and \cite{Lyndon2001}. Let $\pi_1 := \pi_1(L^{1}, q_0)$, and let $S$ be a spanning tree of $L^{1}$. The graph $L^{1}_{H^1}$ has a set of vertices in bijection with $V' := V \times \pi_1/H^1$, and a set of edges in bijection with $E' := E \times \pi_1/H^1$. We now explain how to connect the edges. Let $e = [x, q]$ be an oriented edge from $x$ to $q$ in $L^1$, and let $\gamma^S(e) \in \pi_1$ denote the group element obtained by adding $e$ to the spanning tree $S$. More formally, the equivalence class of loops is denoted
\[
\gamma^S(e) = \left[ \big((q_0, x)\big) \cdot e \cdot \big((q, q_0)\big) \right],
\]
where $\big((a, b)\big)$ is an arbitrary path in $S$ made of oriented edges from vertex $a$ to $b$. Then, the edge $(e, gH^1) \in E'$ refers to the edge $\left[(x, gH^1),\; (q, g \gamma^S(e) H^1)\right]$. This completes the description of the lifted graph $L^1_{H^1}$.\par

We now give an explicit description of the boundary maps of the lifted code. We use the shorthand notation $\gamma^S[x,q] := \gamma^S([x,q])$. The lifted incidence matrices $(\tilde{\delta}^H_1, \tilde{\partial}^H_1)$ can be expressed as follows on basis vectors:
\begin{align}
    \tilde{\delta}^H_1(x, gH^1) &= \sum_{q \in \mathrm{supp}(\tilde{\delta}_1 x)} (\tilde{\delta}_1)_{q,x} \cdot (q, g\gamma^S[x,q] H^1), \nonumber \\
    \tilde{\partial}^H_1(q, gH^1) &= \sum_{x \in \mathrm{supp}(\tilde{\partial}_1 q)} (\tilde{\partial}_1)_{x,q} \cdot (x, g\gamma^S[q,x] H^1), \label{diff1}
\end{align}
where $(\tilde{\delta}_1)_{q,x} \in \mathbb{Z}$ (respectively $(\tilde{\partial}_1)_{x,q}$) denotes the incidence number between vertex $x$ and edge $q$. These maps are extended by linearity over chains.

For the associated CSS code over $\mathbb{F}_2$, we simply take the entries modulo $2$. For instance, the lifted boundary map becomes $\partial^H_1(q, gH^1) = \sum_{x \in \mathrm{supp}(\partial_1 q)} (x, g\gamma^S[q,x] H^1)$.\par

To obtain the space $L_H^2$, we lift the attaching maps of the $2$-cells of $L$ into the $1$-skeleton $L^1_{H^1}$. Proceeding analogously for higher-dimensional cells, we obtain a covering of the $4$-skeleton.\par

We now describe how to compute the lifted boundary maps $(\tilde{\delta}^H_2, \tilde{\partial}^H_2)$ using only the $2$-skeleton of $L$ and $L_H$. Note that lifting the attaching map of each $5$-cell, whose support corresponds to a row of the incidence matrix $\tilde{\partial}^H_2$, completes the construction of the full covering map $p_H: L_H \to L$. Let $z \in Z$, and let $A_z$ denote the subcomplex obtained by adjoining the graph $\widetilde{T}_z$ with the $2$-skeleton of the dressed core $c_z$. Since the $2$-cells in $A_z$ were introduced specifically to kill its fundamental group, $A_z$ is simply connected. As a result, it lifts to $|\pi_1(L)/H|$ disjoint copies in the covering space $L_H$.\footnote{This follows from the fact that the restriction of a covering map to a subcomplex $A \subset L$ is again a covering map. If $A$ is simply connected, the Galois correspondence guarantees that the covering decomposes into disjoint copies of $A$.} Let $S_z$ be a spanning tree of the $1$-skeleton $A_z^1$. Lifting this tree to $p_H^{-1}(S_z)$ allows us to determine the support of each lifted stabilizer, and hence to reconstruct the lifted boundary maps. Choose an arbitrary basepoint $q_z$ in $A_z^1$. Denote by $S_{(z, gH)}$ the unique lift of $S_z$ to $L_H$, based at $(q_z, gH)$, and let $(z, gH)$ be the corresponding lift of the stabilizer $z$. To reach any other vertex (qubit) $q \in A_z^1$ from $q_z$, we follow a path $((q_z, q))$ within $S_z$. The lift of this path, starting from $(q_z, gH)$, terminates at the qubit $(q, g\phi(\gamma^{S_z}(q_z, q))H)$, where the group element $\gamma^{S_z}(q_z, q)$ is defined as
\begin{equation}\label{path Tz}
    \gamma^{S_z}(q_z, q) := \gamma^S[q_z, x_1] \cdot \gamma^S[x_1, q_1] \cdots \gamma^S[x_p, q],
\end{equation}
with the path $((q_z, q))$ decomposed into edges as $[q_z, x_1], [x_1, q_1], \ldots, [x_p, q]$.\par

With this notation, we can express the action of the co-boundary and boundary maps on basis elements as
\begin{align}
    \tilde{\delta}^H_2(q, gH) &= \sum_{z \in \mathrm{supp}(\tilde{\delta}_2 q)} (\tilde{\delta}_2)_{z, q} \Bigl( z,\, g\phi(\gamma^{S_z}(q, q_z))H \Bigr), \nonumber \\
    \tilde{\partial}^H_2(z, gH) &= \sum_{q \in \mathrm{supp}(\tilde{\partial}_2 z)} (\tilde{\partial}_2)_{q, z} \Bigl( q,\, g\phi(\gamma^{S_z}(q_z, q))H \Bigr). \label{diff2}
\end{align}
The parity-check matrices of the associated quantum CSS code are obtained by reducing the coefficients modulo $2$.

\begin{theorem}
The boundary maps defined in (\ref{diff1}) and (\ref{diff2}) satisfy $\tilde{\partial}^H_1\circ\tilde{\partial}^H_2=0$.
\end{theorem}

\begin{proof}
We prove the relation on basis elements and extend it to the general case by linearity. Let $(z, gH)$ be an arbitrary $Z$-check in the cover. Then,
\[
\tilde{\partial}^H_1 \circ \tilde{\partial}^H_2(z, gH) = \sum_{q \in \mathrm{supp}(\tilde{\partial}_2 z)} \sum_{x \in \mathrm{supp}(\tilde{\partial}_1 q)} (\tilde{\partial}_2)_{q, z} (\tilde{\partial}_1)_{x, q} \left( x, g\phi\left( \gamma^{S_z}(q_z, q)\, \gamma^S[q, x] \right) H \right).
\]

Since $S_z$ is fixed for $z$ and the closure of $A_z$ in $L^2$ is simply connected, the element $\gamma^{S_z}(q_z, q)\, \gamma^S[q, x] \in \pi_1(L^1)$ is independent of the choice of $q$. We therefore denote it $\gamma^{S_z}(q_z, x)$. Next, observe that each $x \in \mathrm{supp}(\partial_1 q)$ is also a vertex of $S_z$ by construction. Hence, we may extend the sum to all $x \in S_z$:
\[
\tilde{\partial}^H_1 \circ \tilde{\partial}^H_2(z, gH) = \sum_{x \in S_z} \left( x, g\phi\left( \gamma^{S_z}(q_z, x) \right) H \right) \sum_{q \in \mathrm{supp}(\tilde{\partial}_2 z)} (\tilde{\partial}_1)_{x, q} (\tilde{\partial}_2)_{q, z}.
\]
We conclude by applying the relation $\tilde{\partial}_1 \circ \tilde{\partial}_2 = 0$.
\end{proof}
The new lifted code is obtained by truncating the cellular chain complex of $L_H$ and keeping only the portion
\[
C^{H} := \mathcal{C}_5(L_H, \mathbb{Z}_2) \xrightarrow[]{\partial^H_5} \mathcal{C}_4(L_H, \mathbb{Z}_2) \xrightarrow[]{\partial^H_4} \mathcal{C}_3(L_H, \mathbb{Z}_2).
\]
This code, called a \emph{cellular lift} of $C$, has a number of physical qubits given by $|Q| \cdot |\pi_1(L)/H|$, and similarly for the number of $X$- and $Z$-checks.\par

The chain complex $C^H$ inherits a left action from the automorphism group of the covering, known as the group of deck transformations, and denoted $\operatorname{Deck}(p_H)$, and the dual complex enjoys a corresponding right action by $\operatorname{Deck}(p_H)$. When $H$ is a normal subgroup of $\pi_1(L)$, the group of deck transformations is given by $\operatorname{Deck}(p_H) = \pi_1(L)/H$, and its action is both free and transitive.\par

Additionally, the construction of the cellular lift can be formulated in terms of fiber bundles and balanced products, as described in Sections 3.5 and 3.6 of \cite{GuemardLiftIEEE}.

\subsection{Relation between the cellular-lift and the Tanner-lift}\label{section Relation between the cellular-lift and the Tanner-lift}
In \cite{GuemardLiftIEEE}, a two-dimensional simplicial representation of a code $C$, called the \emph{Tanner cone-complex} and denoted $\mathcal{K}(C)$, was introduced for the purpose of performing code lifting. As a result, we now have two distinct representations of a code: the Tanner cone-complex $\mathcal{K}(C)$ and a cellular realization $L \in \mathcal{L}(C)$, both of which support code lifting constructions. In this section, we discuss the relationship between the structure of the Tanner cone-complex and the cellular realization of a code, and we identify conditions under which the two lifting techniques are equivalent.

The Freedman–Hastings construction proceeds by building a manifold cell-by-cell. As previously discussed, adding a 5-cell, corresponding to a $Z$-check, involves trivializing the fundamental group of the graph $G_z$ by attaching 2-cells along a fundamental cycle basis. This construction is illustrated in Figure~\ref{fig: 5-handle}. \par

In the standard Tanner graph of the code, denoted $\mathcal{T}(C)$ in \cite{GuemardLiftIEEE}, the corresponding subgraph induced by a $Z$-check $z$ is denoted $\mathcal{T}_z \hookrightarrow \mathcal{T}(C_X)$. In the Tanner cone-complex construction, this fundamental group is trivialized by attaching the cone over the subgraph, $\mathrm{C} \mathcal{T}_z$, resulting in the new simplicial complex $\mathcal{T}(C_X) \cup \mathrm{C} \mathcal{T}_z$. This process is detailed in Proposition 3.6 of \cite{GuemardLiftIEEE}. We thus observe that trivializing the fundamental group of $G_z$ in the cellular realization is analogous to trivializing the fundamental group of $\mathcal{T}_z$ in the Tanner cone-complex.\par

However, it is important to note that $G_z$ may have more connected components than $\mathcal{T}_z$. In such cases, the 1-handles connecting these components act as generators of the fundamental group of the cellular realization. Assuming that $\widetilde{T}_z$ is connected and that the $\mathbb{Z}$-lift is support-preserving, attaching 2-cells to eliminate these generators ensures that the fundamental group of the cellular realization becomes isomorphic to that of the Tanner cone-complex. We summarize the correspondence between the structure of the Tanner cone-complex and the cellular realization of a code in Table \ref{tab:Tanner-cone complex vs Cellular realization}.

\begin{table}[]
    \centering
    \begin{tabular}{c|c}
         Tanner cone-complex $\mathcal K(C)$&  Cellular realization $L\in \mathcal L(C)$\\ \hline
         Tanner graph $\mathcal T(C_X)$& 
    $\mathbb Z$-lifted Tanner graph of $C_X$, $\widetilde T$\\
 induced subgraph $\mathcal T_z$&induced subgraph $\widetilde T_z$\\
 cone $\mathrm C \mathcal T_z$&2-skeleton of $h_z$\\
 2-cells of $\mathrm C \mathcal T_z$&2-cells trivializing $\pi_1(G_z)$\\\end{tabular}
    \caption{Correspondence between the structure of the Tanner cone-complex and of the cellular realization of a code $C=\mathrm{CSS} (C_X,C_Z)$}
    \label{tab:Tanner-cone complex vs Cellular realization}
\end{table}

Nevertheless, there are key distinctions between the two approaches to representing a code. Firstly, a fundamental difference arises between taking the cone of a subgraph and trivializing its fundamental group via a chosen fundamental cycle basis, as required for the 5-handle attachment. While the coning operation produces a contractible space, thus introducing no higher homotopy or homology, the process of trivializing the fundamental group by adding 2-cells may introduce higher-dimensional homotopy groups and nontrivial second homology. Secondly, an essential aspect in comparing the two code lifting techniques is the existence of a \textit{support-preserving} $\mathbb{Z}$-lift. It is important to note that not all codes admit such a lift. In Appendix~\ref{section A CSS code admitting no support preserving Z-lift}, we review an example of a CSS code which does not admit a support-preserving $\mathbb{Z}$-lift,, originally presented in~\cite{MathoverflowZ-Lift167615} and merely reformulated in the language of coding theory.\par

When a support-preserving $\mathbb{Z}$-lift is available, the two lifting procedures become equivalent, as the fundamental group of the cellular realization can be made isomorphic to that of the Tanner cone-complex. This equivalence holds for many well-known constructions, including codes derived from regular cellulations of manifolds, hypergraph-product codes~\cite{Tillich2009}, lifted-product codes~\cite{Panteleev2020,Panteleev2021}, and fiber bundle codes~\cite{Hastings2020}. In such cases, we can leverage the cellular-lift in the following way.

\begin{lemma}
Let $C$ be a CSS code admitting a support preserving $\mathbb Z$-lift, and let $C'$ be a Tanner-lift of $C$. Then, there exists a support preserving $\mathbb{Z}$-lift of $C'$.
\end{lemma}
\begin{proof}
    This is a direct consequence of the isomorphism between the fundamental group of $L\in \mathcal{L}(C)$ and $\mathcal{T}(C)$ (the Tanner cone-complex). We can therefore find a cellular-lift of $C$ equivalent to any Tanner-lift of $C$. The lifted code is first given with $\mathbb Z$ coefficients, following the procedure of Section \ref{section definition of the cellular-lift}.
\end{proof}

For the class of codes discussed above finding a support-preserving $\mathbb{Z}$-lift is a straightforward task, as outlined in Section~\ref{Section Application: cellular-lifts of hypergraph-product codes} for the case of hypergraph-product codes (HPC) and lifted-product codes (LPC). However, for other types of codes, such as those constructed by lifting a small base code, as introduced in~\cite{GuemardLiftIEEE}, identifying such a lift can be more challenging. In these cases, determining whether a support-preserving $\mathbb{Z}$-lift exists for the base code may require a nontrivial analysis of the code structure. Nevertheless, if a support-preserving $\mathbb{Z}$-lift does exist for the base code, then the procedure described in Section~\ref{section definition of the cellular-lift} provides an efficient and systematic method to compute a corresponding lift for any of its Tanner-lift.\par

However, when such a support-preserving $\mathbb{Z}$-lift does not exist, the fundamental group of the Tanner cone-complex may differ from that of the corresponding cellular realization. For instance, consider the case where the $\mathbb{Z}$-lifted matrix $\widetilde H_Z$ contains a row, corresponding to a check $z \in Z$, that has support on every qubit, while $ H_Z$ is sparse. In this situation, the associated graph $\widetilde{T}_z$ coincides with the full lifted Tanner graph $\widetilde{T}$. Trivializing the fundamental group of $\widetilde{T}_z$ would then amount to trivializing the fundamental group of the entire cellular realization of the code, thereby obstructing the possibility of performing a nontrivial code lift.\par

To summarize, the family of codes obtained through the cellular-lifting procedure of a CSS code crucially depends on several choices made during the construction of its cellular realization. These include the choice of a $\mathbb{Z}$-lift and the specific 2-cells used to attach the 5-cells. In contrast, these choices are absent in the Tanner-cone complex, making it a more canonical object for the purpose of code lifting.

\subsection{Classification of lifts}\label{section Classification of lifts}

In this section, we intend to bring a partial classication of the codes that can be obtained by cellular-lifting a CSS code $C$. As discussed in the previous section,  the result of code lifting depends on the 
specific cellular realization considered. However, we show that two cellular realizations obtained from identical $\mathbb Z$-lift and with identical (connected) 1-skeletons yield the same set of codes by cellular-lifting. This is captured in the following theorem, where we recall that all elements of a set $\mathcal L(C)$ are obtained from an identical $\mathbb Z$-lift.

\begin{theorem}\label{theorem invariance of codes}
Given two cellular realizations $L_1$ and $L_2$ of $\mathcal L (C)$ with identical 1-skeletons, to any finite cover of $L_1$ corresponds one of $L_2$ such that the resulting codes are equal.
\end{theorem}
By symmetry of the statement, it is sufficient to consider connected coverings over a single cellular realization of $C$ to obtain all its lifts.\par

The key step in the proof of Theorem \ref{theorem invariance of codes} is the invariance of the fundamental group under the choice of 2-cells. Before stating the result, we observe that $L_1 $ and $L_2$ have the same $1$-skeletons, but for each $z\in Z$, the cycle bases chosen to trivialize the fundamental group of the subcomplex $\widetilde T_z\hookrightarrow L_i^1$, for $i=1,2$,  are different. We consider the complex composed of the embedded graph $\widetilde T_z$, together with the 2-cells added to make it simply-connected, and denote it by $\overline{ T}_{i,z}\hookrightarrow L_i^2$.

\begin{lemma}\label{theorem invariance of fundamental group}
Let $L_1$ and $L_2$ be two elements of $\mathcal L(C)$ with identical 1-skeletons. Then, $\pi_1(L_1)\cong\pi_1(L_2)$.
\end{lemma}

To prove this lemma, we will repeatidly use the following proposition, stated without proof, which follows from Van Kampen Theorem.

\begin{proposition}[\cite{HatcherTopo}, Proposition 1.26]\label{Proposition consequence of Van Kampen}
\begin{enumerate}
    \item[(i)] Let $A$ and $X$ be a $2$ dimensional cell complexes, $X$ obtained from $A$ by attaching a set of $2$-cells $\{c_k\}$, with a common basepoint $x_0$. Then the inclusion map $(A,x_0)\hookrightarrow (X,x_0)$ induces a surjective homomorphism $\pi_1(A,x_0)\rightarrow\pi_1(X,x_0)$ whose kernel is the normal subgroup $N\trianglelefteq \pi_1(X,x_0)$ generated by the loops induced by $\{\partial c_k\}$ based at $x_0$ via a path in $A$. 
    \item[(ii)] If $X$ is obtained from $A$ by attaching $n$-cells for a fixed $n > 2$, then the inclusion $A \hookrightarrow X$ induces an isomorphism $\pi_1(A, x_0) \cong \pi_1(X, x_0).$ 
\end{enumerate}

\end{proposition}

\begin{proof}
We construct a sequence of complexes interpolating between $L_1^2$ and $L_2^2$ by modifying $\overline{T}_{1,z}$ into $\overline{T}_{2,z}$ one $Z$-check at a time.\par
Let $L_1^2(z_1^-)$ be the complex obtained from $L_1^2$ by removing all the 2-cells in $\overline{T}_{1,z_1}$ and $L_1^2(z_1)$ obtained by attaching to $L_1^2(z_1^-)$ the 2-cells according to the configuration in $\overline{T}_{2,z_1}$.

We denote $B_1$ and $B_2$ the two cycle bases that generate the fundamental group of $\widetilde T_z$, and used to obtain $\overline{T}_{1,z}$ and $\overline{T}_{2,z}$. Their image under the homomorphism $\pi_1(\widetilde T_z)\to \pi_1(L_1^2(z_1^-))$, induced by the inclusion, therefore generates the same normal subgroup $N\trianglelefteq \pi_1(L_1^2(z_1^-))$.\par

Since the subcomplex $\overline{T}_{i,z_1}$ are both simply connected, we can use part (i) of Proposition \ref{Proposition consequence of Van Kampen} to obtain the isomorphisms $\pi_1(L_1^2(z_1))\cong \pi_1(L_1^2(z_1^-))/N\cong \pi_1(L_1^2)$. Fixing a spanning tree of $L_1^1$ defines a set of generators for the fundamental group of $L_1$. Then, the isomorphism relates two different presentations of the same group, on the same set of generators.\par

Let $L_1^2(z_1,z_2^-)$ be the complexes obtained from $L_1^2(z_1)$ by removing the 2-cells of $\overline{T}_{1,z_2}$ and let $L_1^2(z_1,z_2)$ be obtained by  attaching to $L_1^2(z_1,z_2^-)$ the 2-cells according to the configuration of $\overline{T}_{2,z_2}$. Repeating the procedure above, we can show that $\pi_1(L_1^2(z_1,z_2))\cong \pi_1(L_1^2(z_1))$.\par

We continue the sequence of isomorphism of fundamental group until we arrive at the complex \[L_1^2(z_1,z_2,\dots,z_{|Z|})= L_2^2, \]
yielding $\pi_1(L^2_2)\cong \pi_1(L^2_1)$. We conclude the result by attaching higher dimensional cells, which do not change the fundamental group by part (ii) in Proposition \ref{Proposition consequence of Van Kampen}.
\end{proof}

We can now prove Theorem \ref{theorem invariance of codes}.

\begin{proof}
Let $p_{1}:\tilde L_{1}\rightarrow L_1$ be a connected cover of $L_1$ of index $t$. We can modify it to an index-$t$ connected cover of $ L_2$ by the following steps.\par

For $z_1\in Z$, the subcomplex $\overline{T}_{1,z_1}$ is simply-connected and we can therefore lift it $\tilde L_{1}^2$ into $t$ disjoint subcomplexes. Then, let $\tilde L_{1}^2(z_1^-)\rightarrow L_1^2(z_1)$ be the covering map obtained by replacing all the 2-cells according to the configuration of $\overline{T}_{2,z_1} $ and lifting them in $\tilde L_{1}^2$. Since each lift of this subcomplex is simply connected, this doesn't change the fundamental group by part (i) of Proposition \ref{Proposition consequence of Van Kampen} and $\pi_1(\tilde L_{1}^2(z_1))\cong \pi_1(\tilde L_{1}^2)$ and the isomorphism is the one in the proof of Lemma \ref{theorem invariance of fundamental group}. Moreover, it does not change the incidence between any element of $p^{-1}_{1}(z_1)$ and its neighboring qubits (4-cells), so that the underlying lifted code is unchanged by this procedure.\par

We continue this procedure for all $z\in Z$ and end up with a valid covering $p_{2}:\tilde L_{2}\rightarrow L_2$ such that $\pi_1(\tilde L_{1})\cong \pi_1(\tilde L_{2})$, equal 1-skeletons $\tilde L_{1}^1=\tilde L_{2}^1$ and equal resulting codes.

\end{proof}

For completeness, we expose a stronger relation of homotopy equivalence between cellular realizations verifying an extra condition.

\begin{lemma}\label{homotopy equivalence}
Let $L_1$ and $L_2$ be two cellular realizations of $C$ such that the subcomplexes in $\{\overline{T}_{i,z}: z\in Z\}$, $i=1,2$, are all contractible. Then $L_1 \simeq L_2$.
\end{lemma}

\begin{proof}
Let $L_1^2(z_1)$ be the complex  obtained from $L_1^2$ by removing all the 2-cells in $\overline{T}_{1,z_1}$, and reorganizing them according to the configuration of $\overline{T}_{2,z_1}$. Assuming that each $\overline{T}_{i,z}$ is a contractible subcomplex,  $L_1^2\rightarrow L_1^2/\overline{T}_{1,z_1}$ and  $L_1^2(z_1)\rightarrow L_1^2(z_1)/\overline{T}_{2,z_1}$ are homotopy equivalences\footnote{Proposition 0.17 from Hatcher states that for a CW pair $(X, A)$ consisting of a complex $X$ and a contractible subcomplex $A$, the quotient map $X\rightarrow X/A$ is a homotopy equivalence.}. But $L_1^2/\overline{T}_{1,z_1}=L_1^2(z_1)/\overline{T}_{2,z_1}$, and therefore $L_1(z_1)\simeq L_1$.\par

Similarly, let $L_1^2(z_1,z_2)$ be the complex obtained from $L_1^2(z_1)$ by removing all the 2-cells in the closure of $\overline{T}_{2,z_1}$, and reorganizing them according to the configuration of $\overline{T}_{2,z_2}$. Repeating the procedure above we can show that $L_1(z_1,z_2)\simeq L_1(z_1)\simeq L_1$.\par

We continue the sequence of homotopy equivalence until we arrive at the complex \[L_1(z_1,z_2,\dots,z_{|Z|})= L_2, \]
yielding the result $L_2\simeq L_1$.
\end{proof}

\subsection{Application: cellular-lifts of hypergraph-product codes}\label{Section Application: cellular-lifts of hypergraph-product codes}

As previously noted, a non-trivial cellular-lift of a quantum CSS code can only exist if the associated cellular realization is not simply connected. Consequently, in practice, one should seek sparse $\mathbb{Z}$-lifts, and ideally support-preserving $\mathbb{Z}$-lifts, as discussed in Section~\ref{section Relation between the cellular-lift and the Tanner-lift}. In this section, we analyze our construction in the context of hypergraph product codes (HPCs), which are known to admit support-preserving $\mathbb{Z}$-lifts.\par

The main objective is to prove Theorem~\ref{Theorem classification HPC}, which classifies the cellular-lifts of any HPC. As discussed in Section~\ref{section Relation between the cellular-lift and the Tanner-lift}, the existence of support-preserving $\mathbb{Z}$-lifts implies that the cellular-lift construction is equivalent to the Tanner-lift construction, and the latter were classified in~\cite{GuemardLiftIEEE}. \par

Despite this equivalence, we choose to present a detailed, self-contained approach that relies exclusively on the geometric object and methods introduced in Sections~\ref{section Cellular realization of a code} and~\ref{section definition of the cellular-lift}.\par

In the rest of this section, $C:=C_1\otimes C_2^*$ denotes a HPC constructed from two classical codes $C_1$ and $C_2$ specified by their parity-check matrices $H_1$ and $H_2$ and represented by their Tanner graphs $T_1$ and $T_2$, respectively. Then, $C$ has binary parity-check matrices \[H_X=\left[H_1\otimes \mathbb 1 | \mathbb 1 \otimes H_2^T\right ], \qquad H_Z=\left[\mathbb 1 \otimes  H_2|  H_1 ^T \otimes   \mathbb 1\right ].\]
We consider the following $\mathbb{Z}$-lifted matrix $\tilde{H}_X$, obtained by replacing each entry $0$ or $1$ in the binary matrix $H_X$ with $0$ or $1$ in $\mathbb{Z}$, respectively\footnote{This is referred to as the \emph{naive lift} in~\cite{Freedman2020CSS_Manifold}.}. We proceed similarly for the $\mathbb{Z}$-lifted matrix $\tilde{H}_Z$, except that each $1$ in the right block is replaced with $-1$. Informally, the lifted matrices are given by
\[
\tilde{H}_X = \left[ H_1 \otimes \mathbb{1} \,\middle|\, \mathbb{1} \otimes H_2^T \right], \qquad 
\tilde{H}_Z = \left[ \mathbb{1} \otimes H_2 \,\middle|\, - H_1^T \otimes \mathbb{1} \right],
\]
where the entries are understood to lie in $\mathbb{Z}$. Note that this lift satisfies $\tilde{H}_X \tilde{H}_Z^T = 0$, ensuring the chain complex condition. Furthermore, this $\mathbb{Z}$-lift is support-preserving. Moreover, $\mathcal L(C)$ is the set of all cellular realizations obtained from this specific $\mathbb Z$-lift.

\begin{theorem}\label{Theorem classification HPC}
The cellular-lift of $C$ are classified by the subgroups of $\pi_1(T_1)\times \pi_1(T_2)$.
\end{theorem}

Let us recall that an equivalent version of this theorem, formulated in terms of coverings over the Tanner cone-complex, was used in~\cite{GuemardLiftIEEE} to demonstrate that the asymptotically good quantum LDPC codes introduced in~\cite{Panteleev2020, Panteleev2021} can be obtained as lifts of a small hypergraph product code (HPC). In the same spirit, we derive a similar corollary in the context of cellular-lifts. The proof closely follows that of Corollary~4.8 in~\cite{GuemardLiftIEEE}.

\begin{corollary}
There exist a code $C$ (an HPC), a cellular realization $L\in\mathcal{L}(C)$, and an explicit family $\{C_n\}_\mathbb{N}$ with parameter $[[n,\Theta(n),\Theta(n)]]$ which is obtained by generating coverings over $L$.
\end{corollary}

The proof of Theorem \ref{Theorem classification HPC} mainly relies on the following Lemma.

\begin{lemma}\label{Lemma retract on K}
There exist $L\in\mathcal{L}(C)$, and a complex $K$, such that $L^2$ and  $T_1\times T_2$ both retract onto $K$ via an explicit retraction. In particular, $L^2$ is homotopy equivalent to $T_1\times T_2$.
\end{lemma}
\begin{figure}[t]
  \centering
  \includegraphics[scale=0.8]{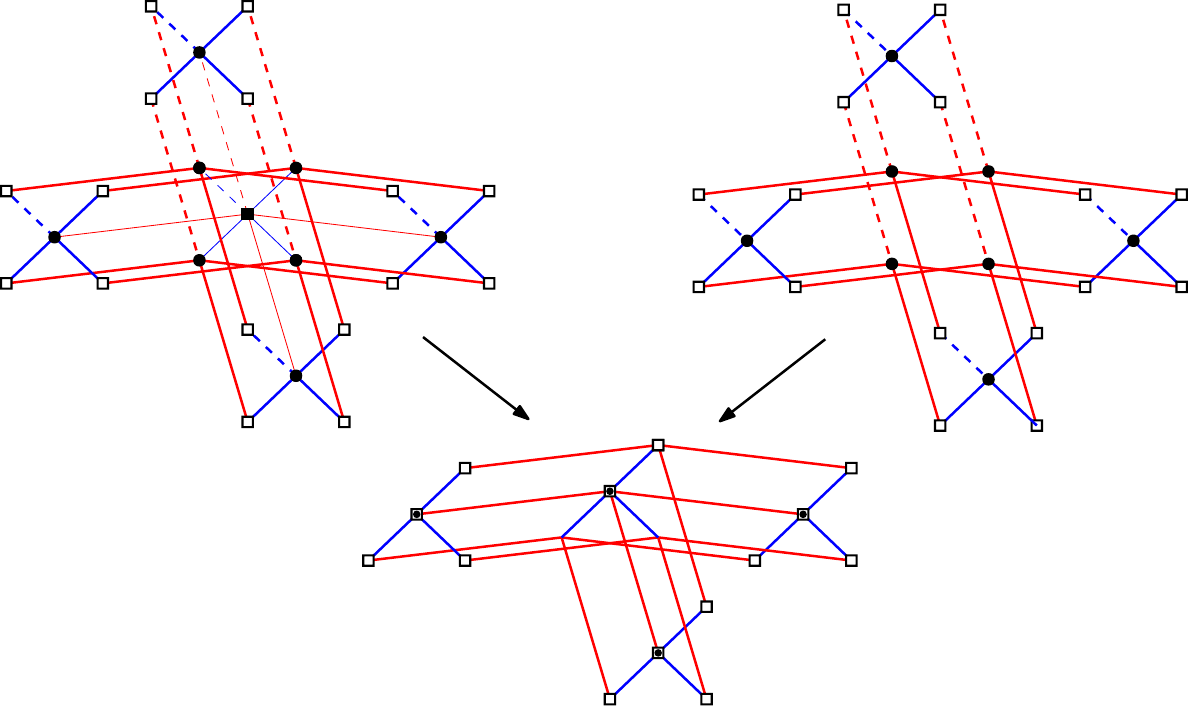}
  \caption{The top left figure represents a local view (the closure of the star) of a $z$ vertex in the complex $T_1\times T_2$, The top right figure represents a local view $\overline{T}_z$ of a $z$ cell in $L\in \mathcal{L}(C)$, where $L$ is the complex considered in Lemma \ref{Lemma retract on K}. Black square, black disks and empty square correspond respectively to $Z$ vertex, bit vertices and $X$ vertices. Red edges and blue edges correspond respectively to horizontal edges (those of $T_1$ )and vertical edges (Those of $T_2$). The thin edges correspond to $Z$-type edges that do not exist in $L$. The black arrow is the retract of the dashed edges, defining a deformation retract of the subcomplex onto a common subcomplex $K_z$ of $K$. The choice of dashed edges results from a choice of spanning forest as described in the proof of Lemma \ref{Lemma retract on K}. The squares containing black dotes represent a $X$ and a bit vertex that are identified after the retraction.}
  \label{fig:retraction HPC on K}
\end{figure}
\begin{proof}
Throughout, by \emph{retraction}, we mean a \emph{deformation retraction}. We outline the steps required to deformation retract both the 2-skeleton $L^2$ and the product complex $T_1 \times T_2$ onto a common subcomplex $K$.\par

We begin by observing that there exists a local configuration of 2-cells in $L^2$ such that the edges of both $L^2$ and $T_1 \times T_2$ can be locally retracted, yielding locally identical subcomplexes. This is illustrated in Figure~\ref{fig:retraction HPC on K}. In order to retract an edge, its closure must be simply connected. However, after performing several local retractions, this property may no longer hold globally. Consequently, we must explain how to make a consistent global choice of edges in order to define a valid deformation retraction across the entire complex.\par

We begin with the product complex $T_1 \times T_2$. Let $(E_i,V_{ib},V_{ic})$ denote the set of edges, bit vertices, and check vertices in $T_i$, for $i = 1, 2$. We define the set of \emph{horizontal edges} of $T_1 \times T_2$ as
\[
E_\rightarrow := E_1 \times V_2,
\]
and the set of \emph{vertical edges} as
\[
E_\uparrow := V_1 \times E_2.
\]
Recall that $C_2^*$ is the code obtained from $C_2$ by interchanging checks and bits. As such, a $Z$-vertex in $T_1 \times T_2$ is an element of $V_{1b} \times V_{2c}$, and an $X$-vertex is an element of $V_{1c} \times V_{2b}$.\par

To construct a valid retraction, we proceed as follows. For each check vertex in $T_1$, select one adjacent edge, and for each bit vertex in $T_2$, select one adjacent edge, such that the corresponding sets of selected edges, denoted $E_1^{\operatorname{f}} \subset E_1$ and $E_2^{\operatorname{f}} \subset E_2$, do not share endpoints in $T_1$ and $T_2$, respectively. These marked edges are to be contracted in a subsequent step to construct the desired deformation retraction.\par

Then for each $z$ vertex $(v_1,v_2)\in V_{1b}\times V_{2c}$, there exists $e\in E_1^{\operatorname{f}}$ such that for all $v\in N(v_2)\cup (v_1,v_2)$, the horizontal edge $(e , v)\in \overline{\operatorname{star}}(v_1,v_2)$, where $N(v_2)$ corresponds to the set of vertices adjacent to $v_2$ in $V_{2b}$, and $\overline{\operatorname{star}}$ is the closure of the star. Similarly, there exists $e\in E_2^{\operatorname{f}}$ such that for all $v\in N(v_1) \cup (v_1,v_2)$, the vertical edge $(v , e)\in \overline{\operatorname{star}}(v_1,v_2)$, where $N(v_1)$ corresponds to the set of vertices adjacent to $v_1$ in $V_{1c}$. In other words, we have a forest of $T_1 $ and $T_2$, such that all edges of the form $ E_1^{\operatorname{f}}\times V_2$ and $V_1\times  E_2^{\operatorname{f}}$ in the closure of the star of $(v_1,v_2)$ can be contracted as described on Figure \ref{fig:retraction HPC on K}. We recall that the product of retractions gives a retraction of the product and we hence obtain a deformation retract from $T_1\times T_2$ to a complex $K$.\par

Let $L^1 = (T_1 \times T_2)^1 \setminus \big( V_{1b} \times V_{2c} \cup V_{1b} \times E_2 \big)$. We can complete $T_z$ into a 2-dimensional complex $\overline{T}_z$ by attaching each face along a simple path formed of 8 edges, which must go through one horizontal edge of $E^{\text{f}}_1 \times V_{2b}$ and one vertical edge of $V_{1b} \times E^{\text{f}}_2$. This determines the path completely, and this constraint forces the boundary to go through exactly two parallel copies of the marked vertical and horizontal edges. It is possible to attach exactly $(|N(v_1)| - 1) \cdot (|N(v_2)| - 1)$ faces of this form on $T_z$, such that the retraction of marked edges forms a simply connected subcomplex $K_z$ as described in Figure \ref{fig:retraction HPC on K}, and hence $\overline{T}_z$ is itself simply connected. Repeating this construction of $\overline{T}_z$ for evert $z \in Z$ and retracting horizontal edges of the form $E^{\text{f}}_1 \times V_2$ and vertical edges of the form $V_1 \times E^{\text{f}}_2$ also yields the complex $K$. This concludes that $T_1 \times T_2 \simeq L^2$.
\end{proof}

\begin{example}
We suggest that the reader apply this procedure in the case where $C$ is the toric code~\cite{Kitaev2003}. In this setting, the 2-skeleton of the complex $L \in \mathcal{L}(C)$ corresponds to the standard cellulation of the torus, but with an additional vertex inserted for each qubit, effectively splitting each edge into two. The complex $T_1 \times T_2$ exhibits a product structure and the simplest spanning forest has on edge out of two in both $T_1$ and $T_2$.
\end{example}

\begin{lemma}\label{lemma fundamental group of product}
Let $L\in\mathcal{L}(C)$ be any cellular realization of $C$. Then it satisfies  $\pi_1 (L)\cong\pi_1(T_1)\times\pi_1(T_2)$.
\end{lemma}
\begin{proof}
There exists a cellular realization of $C$ homotopy equivalent to $T_1 \times T_2$, therefore having a fundamental group isomorphic to $\pi_1(T_1) \times \pi_1(T_2)$. By Theorem \ref{theorem invariance of fundamental group}, any other cellular realization has an isomorphic fundamental group.
\end{proof}

We are now ready for the proof of Theorem \ref{Theorem classification HPC}.

\begin{proof}
Let $L \in \mathcal{L}(C)$ be any cellular realization of $C$. By Lemma \ref{lemma fundamental group of product}, the fundamental group of $L$ satisfies the isomorphism $\pi_1(L) \cong \pi_1(T_1) \times \pi_1(T_2)$. From the Galois correspondence, Theorem \ref{Theorem Galois correspond}, it follows that coverings of $L$ are classified by subgroups of the fundamental group. Finally, by the invariance result, Theorem \ref{theorem invariance of codes}, all lifts of $C$ can be obtained from coverings over any single cellular realization in $\mathcal{L}(C)$. Therefore, subgroups of $\pi_1(T_1) \times \pi_1(T_2)$ provide a complete classification of such lifts.
\end{proof}

\section*{Acknowledgement}

The author would like to thank Benjamin Audoux and Anthony Leverrier for suggesting the investigation of manifold coverings and for their insightful discussions throughout this work. We also thank Matthew Hastings and Guanyu Zhu for their valuable feedback. We acknowledge the Plan France 2030 through the project NISQ2LSQ ANR-22-PETQ-0006.

\bibliographystyle{alpha}
\newcommand{\etalchar}[1]{$^{#1}$}

\appendix

\section{A CSS code admitting no support preserving $\mathbb Z$-lift}\label{section A CSS code admitting no support preserving Z-lift}

In this section, we present a code that does not admit a support-preserving $\mathbb{Z}$-lift. This example, originally given in~\cite{MathoverflowZ-Lift167615}, uses finite projective geometry. Here, we simply translate it into the language of coding theory. Background on quantum codes and projective geometry can be found in~\cite{Audoux2015}.\par

The CSS code we consider is composed of two linear codes, denoted $C_X$ and $C_Z$. The code $C_Z$ is defined by the point-line incidence structure of the projective plane $\mathbb{P}^2(\mathbb{F}_2)$. It is identified with the complex
\[
C_Z := \mathbb{F}_2^7 \xrightarrow[]{H_Z} \mathbb{F}_2^7,
\]
where the first vector space has a basis indexed by the points of the Fano plane, the second by its lines, and $H_Z$ is the (rank 4) line-point incidence matrix that maps each point to the three lines it lies on, i.e. \[H_Z=\begin{bmatrix}
1&0&1&0&0&0&1\\
0&1&0&0&1&0&1\\
1&0&0&0&1&1&0\\
0&1&1&0&0&1&0\\
1&1&0&1&0&0&0\\
0&0&1&1&1&0&0\\
0&0&0&1&0&1&1\\
\end{bmatrix}.\]
To define the code $C_X$, first notice that the seven points of the Fano plane may be labeled by the seven non-zero elements of $\mathbb{F}_2^3$, via a map denoted $\phi$. This can be done in such a way that for every two points $u$ and $v$, the third point $w$ of the line incident to $u$ and $v$ has label $\phi(w) = \phi(u) + \phi(v)$. The code $C_X$ is the linear code $\mathbb{F}_2^7 \xrightarrow[]{H_X} \mathbb{F}_2^3$, where $H_X$ is the matrix representation of $\phi$. Interestingly, $H_X$ is the parity-check matrix of the Hamming code.

\[ H_X=\begin{bmatrix}
1 & 1 & 1&0&1&0&0\\
1&1&0 & 0 & 0&1&1\\
1&0&1&1&0&1&0\\
\end{bmatrix}.\]

From the property of the labeling map $\phi$, the composition of the parity-check matrix yields $H_XH_Z^T=0$, and the chain complex $C:=\mathbb{F}_2^7 \xrightarrow[]{H_Z^T} \mathbb{F}_2^7 \xrightarrow[]{H_X} \mathbb{F}_2^3$ represents a valid CSS code.\par
Suppose there exists a support-preserving $\mathbb{Z}$-lift $\mathbb{Z}^7 \xrightarrow[]{\widetilde{H}_Z^T} \mathbb{Z}^7 \xrightarrow[]{\widetilde{H}_X} \mathbb{Z}^3$ of $C$. Then there exist 7 vectors in $\mathbb{Z}^3$, the images of the basis vectors of the second $\mathbb{Z}^7$, such that each line of the Fano plane gives three vectors with a linear relation between them. That is, the Fano plane would be realized in $\mathbb{P}^2(\mathbb{Q})$, which is not possible. Therefore, a support-preserving $\mathbb{Z}$-lift of $C$ does not exist.

\end{document}